\documentclass[letterpaper, 10 pt, conference]{ieeeconf}  
\usepackage{amsmath,amssymb,amsfonts}
\usepackage{algcompatible}

\usepackage{graphicx}
\usepackage{textcomp}

\IEEEoverridecommandlockouts                              

\let\labelindent\relax 
\usepackage{enumitem}
\usepackage{algorithm}
\usepackage{caption}
\usepackage{subcaption}
\usepackage{hyperref}
\usepackage{siunitx} 
\usepackage{soul}
\usepackage{tabularx}

\newcommand{\ssreview}{\textcolor{black}}
\newcommand{\ssminor}{\textcolor{black}}
\newcommand{\ssupdated}{\textcolor{black}}
\newcommand{\ssupdatedd}{\textcolor{black}}

\usepackage{amsthm}

\newtheorem{theorem}{Theorem}
\newtheorem{corollary}[theorem]{Corollary}
\newtheorem{lemma}[theorem]{Lemma}
\newtheorem{proposition}[theorem]{Proposition}

\newtheorem{remark}{Remark}

\newcommand{\abs}[1]{| #1 |}
\newcommand{\norm}[1]{\left\lVert #1 \right\rVert}
\DeclareMathOperator*{\argmin}{arg\,min} 
\DeclareMathOperator*{\argsup}{arg\,sup} 
\DeclareMathOperator*{\arginf}{arg\,inf} 

\usepackage{mathtools}

\DeclarePairedDelimiter\floor{\lfloor}{\rfloor}

\mathtoolsset{showonlyrefs=true} 

\title{\LARGE \bf
Towards Fast and Safety-Guaranteed Trajectory Planning and Tracking for Time-Varying Systems
}

\author{Seth Siriya, Mo Chen, and Ye Pu
\thanks{\ssminor{The work of Seth Siriya was supported by the Australian Government Research Training Program (RTP) Scholarship. The work of Ye Pu was supported by the Australian Research Council under Project DE220101527.}}
\thanks{\ssminor{S. Siriya and Y. Pu are with the Department of Electrical and Electronic Engineering, University of Melbourne, Parkville, VIC 3010, Australia (\{ssiriya@student., ye.pu@\}.unimelb.edu.au; +61 390355511). M. Chen is with the School of Computing Science, Simon Fraser University, Burnaby, BC V5A 1S6, Canada (mochen@cs.sfu.ca; +1 7787827198).}
}
}

\begin{document}

\maketitle
\thispagestyle{empty}
\pagestyle{empty}

\begin{abstract}
When deploying autonomous systems in unknown and changing environments, it is critical that their motion planning and control algorithms are computationally efficient and can be reapplied online in real time, whilst providing theoretical safety guarantees in the presence of disturbances. The satisfaction of these objectives becomes more challenging when considering time-varying dynamics and disturbances, which arise in real-world contexts. We develop methods \ssupdated{with the potential to} address these issues by applying an offline-computed safety guaranteeing controller on a physical system, to track a virtual system \ssupdatedd{that} evolves through a trajectory that is replanned online, \ssupdated{accounting for constraints updated online}. 
The first method we propose is designed for general time-varying systems over a finite horizon. Our second method overcomes the finite horizon restriction for periodic systems.
We simulate our algorithms on a case study of an autonomous underwater vehicle subject to wave disturbances.
\end{abstract}





\section{Introduction}

Robotic systems have demonstrated great success in a wide range of applications, including self-driving cars in urban environments \cite{badue2021self}, unmanned autonomous vehicles (UAVs) for disaster response \cite{adams2011survey}, and autonomous underwater vehicles (AUVs) \cite{paduan2009high} for seafloor mapping. As they become more commonplace and integrated into society, it is pertinent that they can efficiently plan trajectories to safely maneuver through a priori unknown and changing environments in real time. One of the major challenges to further widespread adoption of these technologies are complex, time-varying, system dynamics and disturbances, which hinder the safety of planned trajectories. This motivates the need to develop strategies for safe navigation of nonlinear systems with  time-varying dynamics and bounded disturbances capable of handling constraints updated online in real time.

Various approaches to safe planning and control have been considered in recent years. In \cite{Majumdar2017FunnelPlanning:}, a finite library of open-loop trajectories and feedback controllers are precomputed offline, \ssupdatedd{and} then sum-of-squares (SOS) optimization is used to compute a safe \textit{funnel}. Subsequent SOS-based methods remove the finite library requirement to consider either a parameterized class of, or arbitrary, planned trajectories (e.g. \cite{smith2019continuous,singh2018robust,kousik2020bridging,yin2020optimization,schweidel2022safe}). Another method is control barrier functions (CBFs) \cite{Xu2015RobustnessControl}, which provide inequality constraints on control inputs that guarantee the system remains in a safe set. CBFs for \ssupdatedd{low-level} control have been combined with model predictive control (MPC) for \ssupdatedd{high-level} planning in \cite{rosolia2020multi}. Control contraction metrics for \ssupdatedd{low-level} control have been applied to track a nominal reference trajectory computed using MPC, with guarantees that the system remains in a tube around the nominal reference \cite{singh2021safe}, bearing similarities to tube-based MPC \cite{Mayne2011Tube-basedControl}. \ssreview{These results often conservatively \ssupdatedd{consider} worst-case additive disturbances over the state space, motivating \cite{dhar2023disturbance} which proposes a planning strategy that accounts for disturbances parameterized over the state space.} However, \ssreview{all of these} methods \ssupdatedd{have} their own \ssupdatedd{respective} challenges, such as the need for polynomial dynamics in SOS-based methods, the difficulty of searching for CBFs for nonlinear systems, and the online computation of nominal reference trajectories in nonlinear MPC. Moreover, they do not consider time-varying dynamics.

FaSTrack \cite{Herbert2018FaSTrack:Planning, Chen2021FaSTrack:Tracking} is another method that follows a similar strategy to address the challenge of safe navigation of nonlinear systems with online constraint satisfaction and bounded disturbances. It achieves this through the offline precomputation of a safe, optimal tracking controller and associated tracking error bound (TEB) using Hamilton Jacobi (HJ) reachability.
This controller is applied online to track a simple \textit{planning system} evolving through a trajectory computed via an arbitrary computationally efficient planning method, \ssreview{enabling real-time performance}. However, time-varying dynamics are not considered. A \ssupdated{naive} extension to safely handle time-varying dynamics is to uniformly bound the time-varying components, but this approach may be too conservative. Moreover, the issue of when or how to replan to effectively navigate through diverse environments is not thoroughly explored. Different speeds for the planning system are considered in \cite{Fridovich-Keil2018PlanningPlanning} to address this, but more expensive offline computation is required. 

Motivated by these challenges, we aim to develop a method for \ssupdatedd{the} safe navigation of \textit{time-varying} nonlinear systems that \ssupdatedd{satisfy} constraints updated online in real time, despite the presence of bounded disturbances. Our approach is inspired by FaSTrack and similarly consists of an offline computation stage and an online planning and tracking stage. We started this line of work in \cite{Siriya2020Safety-GuaranteedWaves}, which proposed a method for real-time safe trajectory planning of AUVs in plane progressive waves --- a time-varying system. Although we will use AUVs as a motivating example, unlike \cite{Siriya2020Safety-GuaranteedWaves}, in this paper we address the combined replanning and tracking problem for general time-varying systems. In particular, our contributions are as follows:

\begin{enumerate}[leftmargin=*]

\item We propose a trajectory replanning and tracking method for \ssupdatedd{the} safe navigation of \textit{general time-varying} dynamical systems to a goal while satisfying constraints \ssupdatedd{that} are updated online, over a \textit{restricted} time horizon.
\ssupdated{Our strategy shows promise for real-time performance, which is enabled via the precomputation of} a value function and safe, optimal tracking controller offline based on a time-varying model of the dynamics, over this horizon.
We derive theoretical guarantees for online constraint satisfaction and \ssupdatedd{goal-reaching}.
\item Our procedure in 1) allows the planning system to be freely initialized in a variably sized set based on the value function during replanning. This allows the \textit{planning system} \ssreview{trajectory to exhibit discontinuous jumps during replanning, or } `teleport', enabling greater flexibility in performance across diverse environments.
\item We propose a trajectory replanning and tracking method for \textit{periodic} systems that overcomes the time horizon restriction in 1) to achieve constraint satisfaction and goal reaching for \textit{arbitrarily} long horizons despite the HJ precomputation being performed over a short time interval. We theoretically \ssupdatedd{analyze} the method.
\item We demonstrate our methods using a numerical example of an AUV in plane progressive waves. Firstly, we show that our time-varying procedure in 1) successfully achieves obstacle avoidance and goal reaching in a tight, \textit{a priori} unknown environment by updating known obstacle information in real time using sensor data. This is contrasted against an alternate time-invariant procedure based on FaSTrack \cite{Chen2021FaSTrack:Tracking}. Secondly, we demonstrate that flexible reinitialization of the planning system can be exploited for adaptive \ssupdatedd{behavior}. Lastly, we demonstrate the success of the procedure from 3) in a lengthy simulation.
\end{enumerate}


The remainder of the paper is organized as follows: \ssminor{Sec.} \ref{sec:problem-method-overview} \ssminor{contains} our problem formulation. In \ssminor{Sec.} \ref{sec:offline-computation}, the offline stage of our framework is provided. In \ssminor{Sec.} \ref{sec:planning}, the planning procedure for the online stage is detailed. The replanning and tracking framework for time-varying systems is provided in \ssminor{Sec.} \ref{sec:replanning}, alongside a theoretical analysis of its safety guarantees. This is repeated in \ssminor{Sec.} \ref{sec:replanning-periodic} for the case of periodic systems. 
\ssreview{In \ssminor{Sec.}~\ref{sec:auv-example}, we provide an example AUV system model subject to wave disturbances and show how our framework can be implemented for it.}
Conclusions are provided in \ssminor{Sec.} \ref{sec:conclusion}. \ssreview{Proofs of all theoretical results are deferred to Appendix~\ref{sec:proofs}.}

\section{Problem Formulation} \label{sec:problem-method-overview}
In \ssreview{this section}, we describe the problem of navigating a time-varying system to a goal while satisfying constraints updated online. 




Consider a dynamical system --- which we henceforth refer to as the \textit{tracking system} --- governed by the dynamics:
\begin{equation} 
    \dot{s} = f(t,s,u_s,d) \label{eqn:time-varying-tracking-model}
\end{equation}
where $f:[0,T] \times \mathcal{S} \times \mathcal{U}_s \times \mathcal{D} \rightarrow \mathbb{R}^{n_s}$ is the \textit{tracking system model}, $t$ is time, $[0,T]$ is the time interval where $f$ is defined, $s \in \mathcal{S}\subseteq \mathbb{R}^{n_s}$ is the \textit{tracking system state}, $u_s$ is the \textit{tracking system control} taking values in a compact set $\mathcal{U}_s \subseteq \mathbb{R}^{n_u}$, and $d$ is the \textit{disturbance} taking values in a compact set $\mathcal{D} \subseteq \mathbb{R}^{n_d}$. 

Our objective is to control the tracking system starting from some initial state $s_0$ at time $t = 0$, so it satisfies control constraint $\mathcal{U}_s$, remains inside a state constraint set $\mathcal{C} \subseteq \mathcal{S}$ which is updated online in real time and reaches a desired goal region $\mathcal{G} \subseteq \mathbb{R}^{n_s}$, over the interval $t \in [0,T_{\textnormal{run}}]$, \ssreview{where $T_{\textnormal{run}} \leq T$ is the \textit{task horizon} over which the system is controlled}. \ssreview{For example, $\mathcal{U}_s$ for an AUV could be its actuator limits, $\mathcal{C}$ may contain velocity constraints and the complement of the obstacles to be avoided in an environment, and $\mathcal{G}$ could be its desired position and velocity (see Sec.~\ref{sec:auv-example} for a specific example).} This is to be achieved regardless of the disturbance $d \in \mathcal{D}$ affecting the system.

\ssreview{Throughout this paper, we} assume $f$ is uniformly continuous, bounded, and Lipschitz continuous in $s$ uniformly in $t$, $u_s$, and $d$. Therefore, given control inputs $u_s(\cdot) \in \{ \phi : [t,T] \rightarrow \mathcal{U}_s : \phi(\cdot) \text{ is measurable}\}$, disturbances $d(\cdot) \in \{ \phi : [t,T] \rightarrow \mathcal{D} : \phi(\cdot) \text{ is measurable}\}$, and initial state $s \in \mathcal{S}$ at time $t \in [0,T]$, there exists a unique continuous trajectory that solves \eqref{eqn:time-varying-tracking-model} \cite{coddington1955theory}. We denote this trajectory as $\xi_f(\cdot;s,t,u_s(\cdot),d(\cdot)) : [t,T] \rightarrow \mathcal{S}$. It satisfies \eqref{eqn:time-varying-tracking-model} with an initial condition almost everywhere, such that $\frac{d}{dt'} \xi_f(t';s,t,u_s(\cdot),d(\cdot))= f(t',\xi_f(t';s,t,u_s(\cdot),d(\cdot)),u_s(t'),d(t'))$ for $t' \in [t,T]$, with $\xi_f(t;s,t,u_s(\cdot),d(\cdot)) = s$.


\section{Offline Computation} \label{sec:offline-computation}

In this section, we describe the offline stage of our method for trajectory replanning and tracking for time-varying systems. We start by describing the planning and relative systems in \ssminor{Sec.} \ref{sec:planning-relative-system}. Following this, in \ssminor{Sec.} \ref{sec:hj-game} we formally introduce the value function and optimal tracking controller, and their computation. We also provide the sublevel set invariance guarantee associated with the controller and value function, then summarize the offline stage in Alg. \ref{alg:offline-stage}.

\subsection{The Planning and Relative System} \label{sec:planning-relative-system}

\subsubsection{Planning system} \label{sec:planning-model}

The \textit{planning system} is governed by the dynamics:
\begin{align}
    \dot{p} = h(p,u_p), \quad p \in \mathcal{P}, \quad u_p \in \mathcal{U}_p, \label{eqn:time-varying-planning-model}
\end{align}
where $h: \mathcal{P} \times \mathcal{U}_p \rightarrow \mathbb{R}^{n_p}$ is the \textit{planning system model}, $p \in \mathcal{P} \subseteq \mathbb{R}^{n_p}$ is the \textit{planning system state}, $u_p \in \mathcal{U}_p \subseteq \mathbb{R}^{n_{u_p}}$ is the \textit{planning system control}, and $\mathcal{U}_p$ is a compact set representing the planning system control constraint. We assume $h$ is uniformly continuous, bounded, and Lipschitz continuous in $p$ uniformly in $u_p$. Next, define the set of measurable functions $\mathbb{U}_p(t) := \{ \phi : [t,T_{\textnormal{off}}] \rightarrow \mathcal{U}_p \ : \ \phi(\cdot) \textnormal{ is measurable}\}$. 
Here, $T_{\textnormal{off}} \in [0,T]$ denotes the \textit{offline computation horizon}, which is the finite horizon over which the offline computation is to be performed. Then, given planning system controls $u_p(\cdot) \in \mathbb{U}_p(t)$ and initial state $p \in \mathcal{P}$ at time $t \in [0,T_{\textnormal{off}}]$, there exists a unique continuous trajectory that solves \eqref{eqn:time-varying-planning-model} \cite{coddington1955theory}.
We denote this trajectory as $\xi_h(\cdot;p,t,u_p(\cdot)):[t,T_{\textnormal{off}}] \rightarrow \mathcal{P}$. It satisfies \eqref{eqn:time-varying-planning-model} with an initial condition almost everywhere, such that $\frac{d}{dt'} \xi_h(t';p,t,u_p(\cdot)) = h(\xi_h(t';p,t,u_p(\cdot)),u_p(t'))$ for $t' \in [t,T_{\textnormal{off}}]$, with $\xi_h(t;p,t,u_p(\cdot))=p$.

\subsubsection{Relative system} \label{sec:relative-model}

The \textit{relative system model} is defined as
\begin{equation}
    g(t,r,u_s,u_p,d) := Lf(t,s,u_s,d) - M h(p,u_p) \label{eqn:time-varying-relative-model}
\end{equation}
for $t \in [0,T_{\textnormal{off}}]$, $u_s \in \mathcal{U}_s$, $u_p \in \mathcal{U}_p$, $d \in \mathcal{D}$ and relative state $r = Ls-Mp$.
Here, $L$ and $M$ are specially chosen matrices that respectively embed the tracking system and planning system states into the relative system state space {$\mathcal{R} \subseteq \mathbb{R}^{n_r}$, where $n_r$ is the dimensionality of the space.} 
\ssreview{Note that these systems do not need to evolve in the same space, allowing for the choice of simple planning models. However, care} 
{needs to be taken when selecting $M \in \mathbb{R}^{n_r \times n_p}$ and $L \in \mathbb{R}^{n_r \times n_s}$ so that 1) common coordinates for the planning and tracking system are matched, forming the \textit{tracking error}, and 2) important information in the system dynamics is not lost.} The latter is not considered in FaSTrack \cite{Chen2021FaSTrack:Tracking}, but its importance will become apparent \ssupdatedd{in} our example in \ref{sec:case-study-offline-computation}. Now, consider the ODE
\begin{align}
    \dot{r}=g(t,r,u_s,u_p,d). \label{eqn:relative-system}
\end{align}
Define the sets of measurable functions $\mathbb{U}_s(t) := \{ \phi : [t,T_{\textnormal{off}}] \rightarrow \mathcal{U}_s : \phi(\cdot) \textnormal{ is measurable}\}$, and $\mathbb{D}(t) := \{ \phi : [t,T_{\textnormal{off}}] \rightarrow \mathcal{D} : \phi(\cdot) \textnormal{ is measurable}\}$. Under the assumptions imposed on \eqref{eqn:time-varying-tracking-model} and \eqref{eqn:time-varying-planning-model}, given $u_s(\cdot) \in \mathbb{U}_s(t)$, $u_p(\cdot) \in \mathbb{U}_p(t)$, $d(\cdot) \in \mathbb{D}(t)$, and initial relative state $r \in \mathcal{R}$ at time $t \in [0,T_{\textnormal{off}}]$, there exists a unique continuous trajectory that solves \eqref{eqn:relative-system} \cite{coddington1955theory}. 
We denote this trajectory as $\xi_g(\cdot;r,t,u_s(\cdot),u_p(\cdot),d(\cdot)) : [t,T_{\textnormal{off}}] \rightarrow \mathcal{R}$. It satisfies \eqref{eqn:relative-system} with an initial condition almost everywhere:
\begin{align}
    &\frac{d}{dt'} \xi_g(t';r,t,u_s(\cdot),u_p(\cdot),d(\cdot)) = \\
    &\ \ g(t',\xi_g(t';r,t,u_s(\cdot),u_p(\cdot),d(\cdot)),u_s(t'),u_p(t'),d(t')), \ t' \in [t, T_{\textnormal{off}}],\\
    &\xi_g(t;r,t,u_s(\cdot),u_p(\cdot),d(\cdot))=r.
\end{align}
\subsection{The Value Function and Optimal Tracking Controller} \label{sec:hj-game}
We begin by describing the differential game, which is used to define the value function, and subsequently the optimal tracking controller. 
Firstly, we define the error function $l(r)=\norm{Cr}_2$ over the relative state space, where $C$ is a matrix that selects the coordinates of the relative system we wish to \ssupdatedd{penalize} so that $l(r)$ represents the tracking error between the tracking and planning system.

Next, recall that in our relative system, the tracking/planning system controls and disturbances all affect the system. We consider each of these as separate players influencing the outcome of the game. Specifically, it is a pursuit-evasion game \cite{Mitchell2005AGames}, where the tracking system takes actions to minimize the maximum tracking error over the remainder of the game, and the planning system and disturbances try to maximize it. \ssreview{Intuitively, this means the tracking and planning systems take the optimal action to pursue/avoid each other respectively.}

{The strategy the planning system employs is a mapping $\gamma_p : \mathbb{U}_s(t) \rightarrow \mathbb{U}_p(t)$, such that the planning system control is chosen based on the tracking control, and restricted to non-anticipative strategies $\gamma_p \in \Gamma_p(t)$, defined in \cite{Mitchell2005AGames,Fisac2015Reach-avoidConstraints}. Informally, this means that at each time, the planning system reacts to all tracking system controls so far, but not future controls. The disturbance strategy is similarly defined as $\gamma_d:\mathbb{U}_s(t) \rightarrow \mathbb{D}(t)$ with $\gamma_d \in \Gamma_d(t)$.}

We now define the \textit{value function} as the highest cost attained when the tracking system, planning system, and disturbances behave optimally over the interval $t \in [0,T_{\textnormal{off}}]$:
\begin{align} 
    V(r,t) = \ssupdated{\sup_{\substack{\gamma_p\in\Gamma_p(t) \\ \gamma_d\in\Gamma_d(t)}} \inf_{u_s(\cdot)\in\mathbb{U}_s(t)} \nu  (r,t,u_s(\cdot),\gamma_p[u_s](\cdot),\gamma_d[u_s](\cdot))}, \label{eqn:value-function}
\end{align}
\ssupdated{where $\nu(r,t,u_s,u_p,d):=\max_{\tau \in [t,T_{\textnormal{off}}]} l\left( \xi_g \left( \tau;r,t,u_s,u_p,d  \right) \right)$.}
Unlike FaSTrack \cite{Chen2021FaSTrack:Tracking}, the relative system \eqref{eqn:relative-system} corresponding to $\xi_g$ in \eqref{eqn:value-function} is time-varying. Thus, we adopt the forward-time convention from \cite{Fisac2015Reach-avoidConstraints} for the value function to avoid confusion when discussing time-varying systems. It is obtained by solving the Hamilton-Jacobi-Isaacs (HJI) variational inequality
\begin{align}
    \max \Bigg\{ &\frac{\partial V(r,t)}{\partial t} + \min_{u_s\in\mathcal{U}_s} \max_{u_p \in \mathcal{U}_p, d \in \mathcal{D}} \Bigg[ \nonumber \frac{\partial V(r,t)}{\partial r}^{\top} g(t,u_s,u_p,d) \Bigg], \\
    &l(r) - V(r,t) \Bigg\}=0,
    \label{eqn:pde}
\end{align}
over $t\in[0, T_{\textnormal{off}}]$ with $ V(r,T_{\textnormal{off}})=l(r)$. The solution can be numerically approximated using existing tools in HJ reachability (e.g. the helperOC toolbox \cite{hjreachability} or the OptimizedDP library \cite{bui2022optimizeddp}).

The \textit{optimal tracking controller} is obtained as the optimizer of the Hamiltonian as follows:
\begin{equation}
    u_s^*(r,t) \in \argmin_{u_s \in \mathcal{U}_s} \max_{u_p \in \mathcal{U}_d, \ d \in \mathcal{D}} \frac{\partial V(r,t)^{\top}}{\partial r} g(t,r,u_s,u_p,d). \label{eqn:optimal-tracking-controller}
\end{equation}
By applying this controller, the trajectory of the relative system is non-increasing in the levels of the value function, \ssupdated{as} long as the planning system controls and disturbances take values in $\mathcal{U}_p$ and $\mathcal{D}$. To see this, we provide Prop. \ref{prop:remain-in-teb}. 

\begin{proposition} \label{prop:remain-in-teb} (Sublevel Set Invariance) 
Let $s$ and $p$ be the states of the tracking system \eqref{eqn:time-varying-tracking-model} and planning system \eqref{eqn:time-varying-planning-model} at time $t \in [0,T_{\textnormal{off}}]$, and $r = Ls-Mp$ be the corresponding relative state. For all $t' \in [t,T_{\textnormal{off}}]$,
\begin{align}
    V(\xi_g^*(t';r,t), t') \leq V(r,t), \label{eqn:V-non-increasing}
\end{align}
where
\begin{align}
    & \xi_g^*(t';r,t) := \xi_g(t'; r, t, u_s^*(\cdot; r,t), u_p^*(\cdot;r,t),d^*(\cdot;r,t)) \label{eqn:optimal-relative-system-trajectory}
\end{align}
and $u_p^*(\cdot ; r,t) \in \mathbb{U}_p(t)$, $d^*(\cdot;r,t) \in \mathbb{D}(t)$, $u_s^*(\cdot;r,t) \in \mathbb{U}_s(t)$, $\gamma_p^*[\cdot;r,t] \in \Gamma_p(t)$, $\gamma_d^*[\cdot;r,t] \in \Gamma_d(t)$ are defined as:
\begin{align}
    &u_p^*(t';r,t)=\gamma_p^*[u_s^*(\cdot;r,t);r,t](t'), \label{eqn:worst-planning-controls} \\
    &d^*(t';r,t)=\gamma_d^*[u_s^*(\cdot;r,t);r,t](t'),  \label{eqn:worst-disturbances} \\
    &u_s^*(\cdot;r,t) \in \arginf_{u_s(\cdot)\in\mathbb{U}_s(t)} \ssupdated{\nu(r,t,u_s(\cdot),\gamma_p^*[u_s;r,t](\cdot),\gamma_d^*[u_s;r,t](\cdot))}, \label{eqn:optimal-tracking-controls} \\ 
    &\gamma_p^*[\cdot;r,t] \in \argsup_{\gamma_p \in \Gamma_p(t)} \inf_{u_s(\cdot)\in\mathbb{U}_s(t)}  \ssupdated{\nu(r,t,u_s(\cdot),\gamma_p[u_s](\cdot),\gamma_d^*[u_s;r,t](\cdot))} , \nonumber \\
    &\gamma_d^*[\cdot;r,t] \in \argsup_{\gamma_d \in \Gamma_d(t)} \sup_{\gamma_p \in \Gamma_p(t)} \inf_{u_s(\cdot) \in \mathbb{U}_s(t)} \\ & \qquad \qquad \quad \ssupdated{\nu(r,t,u_s(\cdot),\gamma_p[u_s](\cdot),\gamma_d[u_s](\cdot))} . \nonumber
\end{align}
\end{proposition}
Here, $\gamma_p^*[\cdot;r,t]$ and $\gamma_d^*[\cdot;r,t]$ denote the worst-case non-anticipative strategies for the planning system controls and disturbance respectively, $u_s^*(\cdot;r,t)$ denotes the \textit{optimal tracking system controls} assuming that the planning system controls and disturbances are chosen according to those worst-case strategies, and $u_p^*(\cdot;r,t)$ and $d^*(\cdot;r,t)$ denote the corresponding \textit{worst-case planning system controls and disturbances}. Moreover, $\xi_g^*(\cdot;r,t):[t,T_{\textnormal{off}}] \rightarrow \mathcal{R}$ denotes the \textit{optimal relative system trajectory} when the optimal tracking system controls and worst-case planning system controls and disturbance are applied to the relative system \eqref{eqn:time-varying-relative-model}. \ssreview{These worst-case considerations are usually conservative since the planning system is typically not always adversarially escaping the tracking system. However, they guarantee the optimal tracking controller \eqref{eqn:optimal-tracking-controller} achieves sublevel set invariance in $V(r,t)$ (see Prop.~\ref{prop:remain-in-teb}) with arbitrary planning system controls, which we elaborate upon in Remark~\ref{remark:practical-conditions}. This decouples the two so the former can be pre-computed offline and the latter chosen online, which is key to enabling real-time performance since solving \eqref{eqn:pde} for $V(r,t)$ is computationally expensive. Although not our main focus, existing works (see \cite{Fridovich-Keil2018PlanningPlanning}, \cite{sahraeekhanghah2021pa}) \ssupdated{reduce} the conservatism associated with worst-case planning systems. This issue is somewhat orthogonal to our results since their techniques can be applied on top of what we propose.}

\begin{remark} \label{remark:practical-conditions}
Although Prop. \ref{prop:remain-in-teb} requires $u_s^*$ \eqref{eqn:optimal-tracking-controls} for the tracking system controls and $u_p^*$ \eqref{eqn:worst-planning-controls}, $d^*$ \eqref{eqn:worst-disturbances} for the planning system controls and the disturbances, sublevel set invariance should still hold when the optimal tracking controller \eqref{eqn:optimal-tracking-controller} is applied and the planning system controls and disturbances arbitrarily take values in $\mathcal{U}_p$ and $\mathcal{D}$. There are two reasons for this. Firstly, the controls generated by the optimal tracking controller \eqref{eqn:optimal-tracking-controller} with the worst-case disturbances and planning system controls should be equivalent to \eqref{eqn:optimal-tracking-controls}. This is true for \ssupdatedd{single-player} Bolza problems with differentiable value functions \cite{bertsekas2012dynamic}, but should also usually hold in our \ssupdatedd{two-player} minimum cost problem setting with an almost everywhere differentiable value function. Secondly, compared to the worst-case planning system controls $u_p^*$ and disturbances $d^*$, it should be ``easier'' for the optimal tracking controller \eqref{eqn:optimal-tracking-controller} to oppose arbitrary planning system controls and disturbances taking values in $\mathcal{U}_p$ and $\mathcal{D}$ and hence achieve sublevel set invariance.
\end{remark}

\begin{remark} \label{remark:numerical-error}
\ssreview{The guarantee in \ssminor{Prop.}~\ref{prop:remain-in-teb} requires that the optimal tracking controller $u_s^*$ is exactly obtained, which relies on exactly obtaining $V$. In practice, this may not be achieved when numerically solving \eqref{eqn:pde} over a grid (although the numerical PDE solution converges to the exact solution as the grid spacing goes to zero). In this paper, we assume that $V$ is exactly obtained when providing guarantees, and leave consideration of the effects of errors to future work.}
\end{remark}

The overall steps required for the offline computation are summarized in Alg. \ref{alg:offline-stage}. The most computationally expensive procedure is line 3 since dynamic programming suffers from the curse of dimensionality. 
\ssreview{However, recent techniques have shown promise for solving higher dimensional HJ PDEs, such as the decomposition of systems into smaller subsystems \cite{chen2018decomposition}, and the implementation of high-performance solvers \cite{bui2022optimizeddp} --- which was used for the examples in this work --- although this remains an active field of research.}
After successfully computing $V(r,t)$ and $u_s^*(r,t)$, one can apply them online in real time, which will be useful in \ssminor{Sec.} \ref{sec:replanning} and \ref{sec:replanning-periodic}.
\begin{algorithm}
\caption{OfflineStage}
\begin{algorithmic}[1]
    \STATE \textbf{Input:} $L \in \mathbb{R}^{n_r \times n_s}$, $M \in \mathbb{R}^{n_r \times n_p}$, $C$, $f(t,s,u_s,d) \textnormal{ for } t \in [0,T_{\textnormal{off}}]$, $ h(p,u_p)$, $T_{\textnormal{off}} \in [0,T]$
    \STATE Formulate $g(t,r,u_s,u_p,d), t \in [0,T_{\textnormal{off}}]$ using \eqref{eqn:time-varying-relative-model}
    \STATE Compute $V(r,t), t \in [0,T_{\textnormal{off}}]$ and $\frac{\partial}{\partial r} V(r,t), t \in [0,T_{\textnormal{off}}]$ using dynamic programming in \ssminor{Sec.} \ref{sec:hj-game}
    \STATE Construct the optimal safe controller $u_s^*(r,t)$ from \eqref{eqn:optimal-tracking-controller}
    \STATE \textbf{Output:} $V(r,t)$, $ u_s^*(r,t)$ for $ t \in [0,T_{\textnormal{off}}]$
\end{algorithmic}    \label{alg:offline-stage}
\end{algorithm}

\section{Trajectory Planning for Time-Varying Systems} \label{sec:planning}

In this section, we describe the online planning component of our method. We first introduce essential mathematical objects in \ssminor{Sec.} \ref{sec:sets}, before describing our algorithm in \ssminor{Sec.} \ref{sec:planning-bounded-interval}. Subsequently, in \ssminor{Sec.} \ref{sec:planning-analysis}, we theoretically \ssupdatedd{analyze} the properties of our planning method with respect to our objectives of constraint satisfaction and goal-reaching when combined with the optimal tracking controller \eqref{eqn:optimal-tracking-controller}. 
Throughout this section, we assume that the state constraints $\mathcal{C} \subseteq \mathcal{S}$ are static (i.e. they do not change over time).

\subsection{The Planning System Sublevel Set, Set Satisfaction Map, and Set Avoidance Map} \label{sec:sets}

Recall that the key idea from Prop. \ref{prop:remain-in-teb} is that the optimal relative system trajectory \eqref{eqn:optimal-relative-system-trajectory} is non-increasing in the levels of $V$ over time \eqref{eqn:V-non-increasing}. Given the tracking system state $s$ at time $t$, we can find the set of planning system states $p$ such that the relative system is within the sublevel set corresponding to a particular \textit{value level} $c \in \mathbb{R}$ of $V$. This is called the \textit{planning system sublevel set} $\mathcal{T}_p$, and is defined as
\begin{align}
    \mathcal{T}_p(s,t;c) := \{ p : V(Ls-Mp,t) \leq c \} \subseteq \mathcal{P}, \label{eqn:planning-tube}
\end{align}
where $t \in [0,T_{\text{off}}]$. For any given $s$ and $t$, the set {$\mathcal{T}_p(s,t;c)$} is nonempty if $c \geq \underbar{V}(s,t)$, where $\underbar{V}:\mathcal{S}\times[0,T_{\text{off}}]\rightarrow \mathbb{R}$ denotes the \textit{minimum value level function}:
\begin{align}
    \underbar{V}(s,t) := \min_p V(Ls-Mp,t).  \label{eqn:minimum-level-value-function}
\end{align}

Next, let us define the \textit{set satisfaction map} as
\begin{equation}
\mathcal{F}_B(t;c) := \bigcap_{s \in B^{\mathsf{c}}} \mathcal{T}_p^{\mathsf{c}}(s,t;c). \label{eqn:set-satisfaction-map}
\end{equation}
given set $B \subseteq \mathcal{S}$ and $t \in [0,T_{\text{off}}]$ (visualized in Fig. \ref{fig:set-relation}), where $\mathsf{c}$ denotes set complement. It represents the set of planning system states \ssupdatedd{that} guarantee the tracking system state is inside the set $B$, supposing the relative system is in the sublevel set of $V$ corresponding to $c$. It is useful for translating constraints and goal requirements in the tracking system space $\mathcal{S}$ to the planning system space $\mathcal{P}$. This can be seen in Fig. \ref{fig:trajectory-set} for the case where $n_s=n_p$, such that when the planning system is inside $\mathcal{F}_B$, the tracking system is guaranteed to be inside set $B$. Note that $n_s=n_p$ does not need to hold in general.

A related object is the \textit{set avoidance map}, defined as $\mathcal{A}_B(t;c):=\bigcup_{s\in B} \mathcal{T}_p(s,t;c)$ (visualized in Fig. \ref{fig:set-relation}). It is the set of planning system states which, when avoided by the planning system, ensures the tracking system avoids $B$, shown in Fig. \ref{fig:trajectory-set}. This map is useful when it is more natural to think of constraint satisfaction as avoiding an undesirable set $B$ (e.g. obstacles). It can be easily seen that $p \not \in \mathcal{A}_B(t;c)$ is equivalent to $p \in \mathcal{F}_{B^{\mathsf{c}}}(t;c)$, ensuring the tracking system is in $B^{\mathsf{c}}$.

\begin{figure}
     \centering
     \begin{subfigure}[b]{0.22\textwidth}
         \centering
         \includegraphics[width=0.45\textwidth]{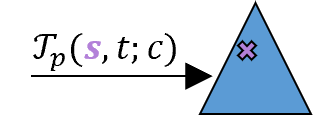}
         \caption{Planning system sublevel set $\mathcal{T}_p$ (\ssreview{blue triangle}) corresponding to a particular tracking system state $s$ (\ssreview{purple `x'}).}
         \label{fig:planning-tube-visual}
         \includegraphics[width=1.0\textwidth]{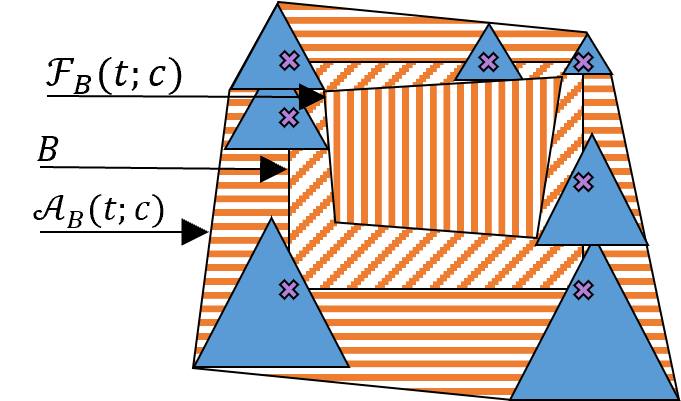}
         \caption{Relationship between $B$ (\ssreview{orange diagonal}), $\mathcal{A}_B$ (\ssreview{orange horizontal}), $\mathcal{F}_B$ (\ssreview{orange vertical}), and $\mathcal{T}_p$ at various tracking system states $s$.}
         \label{fig:set-relation}
     \end{subfigure}
     \hfill
     \begin{subfigure}[b]{0.22\textwidth}
         \centering
         \includegraphics[width=0.9\textwidth]{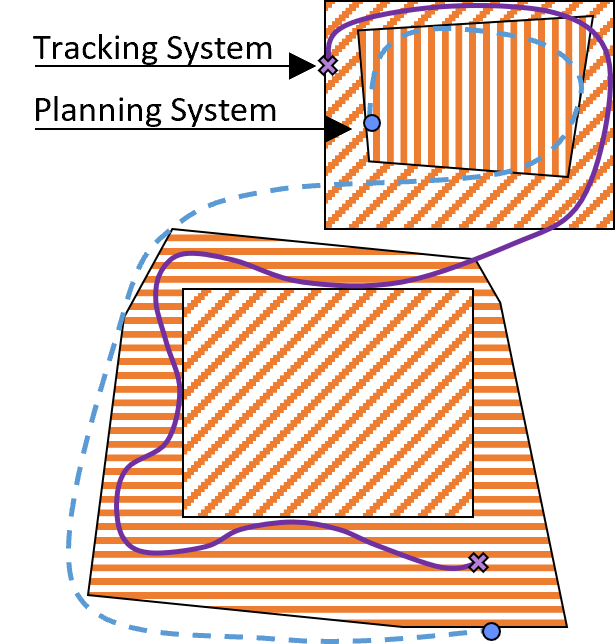}
        \caption{Trajectories showing that when the planning system \ssreview{(blue dashed line)} avoids $\mathcal{A}_B$, the tracking system \ssreview{(purple line)} avoids $B$, and when the planning system is inside $\mathcal{F}_B$, the tracking system is inside $B$.}
         \label{fig:trajectory-set}
     \end{subfigure}
        \caption{\ssminor{Illustrations showing how $\mathcal{T}_p$ is used to construct $\mathcal{F}_B$ and $\mathcal{A}_B$, and their importance for the planning system and tracking system trajectories, when $n_s = n_p$.}}
        \label{fig:set-picture}
\end{figure}

\subsection{Trajectory Planning Method Over a Bounded Interval}
\label{sec:planning-bounded-interval}

We describe our trajectory planning method for time-varying systems over any bounded time interval $[t_{\text{i}},t_{\text{f}}] \subseteq [0,T_{\text{off}}]$, which is summarized in Alg. \ref{alg:planning-bounded-interval}. The inputs to the algorithm include the value function $V(r,t), t \in [0,T_{\text{off}}]$ --- which is obtained from OfflineStage (Alg. \ref{alg:offline-stage}) --- the planning system model \eqref{eqn:time-varying-planning-model}, the goal $\mathcal{G} \subseteq \mathcal{S}$, the planning start time $t_{\text{i}} \in [0,T_{\text{off}}]$, and the planning end time $t_{\text{f}} \in [t_{\text{i}},T_{\text{off}}]$. The static constraints $\mathcal{C} \subseteq \mathcal{S}$ are also supplied as an input. 
Moreover, the initial planning state is an input with $p \in \mathcal{T}_p(s,t_{\text{i}};c)$ enforced, since it is required for establishing safety guarantees. Its importance will be discussed in \ssminor{Sec.} \ref{sec:planning-analysis}. The final input is a chosen value level $c$ satisfying $c \geq \underbar{V}(s,t)$, where $s\in \mathcal{S}$ is the tracking system state at time $t_{\text{i}}$. Although enforcing $c=\underbar{V}(s,t)$ is sufficient for $\mathcal{T}_p(s,t;c)\neq \emptyset$, allowing $c$ to be flexibly chosen enables \ssupdatedd{the} adaptive performance of the controlled system. This will become apparent in \ssminor{Sec.} \ref{sec:replanning}.

Firstly, using $\mathcal{F}_B(t;c)$ \eqref{eqn:set-satisfaction-map}, we compute the  \textit{planning system constraints} $\mathcal{C}_p$ from $\mathcal{C}$, and the \textit{planning system goal} $\mathcal{G}_p$ from $\mathcal{G}$, over the interval $[t_{\text{i}},t_{\text{f}}]$ for the particular level $c$ we have chosen.
Following this, we plan a trajectory $\hat{\xi}_h(t)$ over the interval $t \in [t_{\text{i}},t_{\text{f}}]$, where $\hat{u}_p(\cdot) \in \mathbb{U}_p(t_{\text{i}})$ are the controls for the planned trajectory. The planned trajectory must satisfy Conditions C1-C3 and optionally C4.

\begin{algorithm}[H]
\caption{PlanTrajectoryOverInterval}
\begin{algorithmic}[1]
    \STATE \textbf{Input}: $V(r,t) \text{ for } t \in [0,T_{\text{off}}]$, $ h(p,u_p)$ for $t_{\text{i}} \in [0,T_{\text{off}}]$, $s \in \mathcal{S}$,  $t_{\text{f}} \in [t_{\text{i}},T_{\text{off}}]$, $\mathcal{C} \subseteq \mathcal{S}$, $\mathcal{G} \subseteq \mathcal{S}$, $c$ satisfying $c \geq \underbar{V}(s,t)$, $p$ satisfying $p \in \mathcal{T}_p(s,t_{\text{i}};c)$
    \STATE Compute $\mathcal{C}_p(t) \leftarrow \mathcal{F}_{\mathcal{C}}(t;c)$ for $t \in [t_{\text{i}},t_{\text{f}}]$ from \eqref{eqn:set-satisfaction-map}
    \STATE Compute $\mathcal{G}_p(t) \leftarrow \mathcal{F}_{\mathcal{G}}(t;c)$ for $ t \in [t_{\text{i}},t_{\text{f}}]$ from \eqref{eqn:set-satisfaction-map}
    \STATE Plan trajectory $\hat{\xi}_h(t), t\in[t_{\text{i}},t_{\text{f}}]$ based on $h$ satisfying C1-C3, and optionally C4 (see Remark \ref{remark:conditions})
    \STATE \textbf{Output:} $\hat{\xi}_h(t) \text{ for } t \in [t_{\text{i}},t_{\text{f}}]$
\end{algorithmic}    \label{alg:planning-bounded-interval}
\end{algorithm}

\begin{enumerate}[label=C\arabic*.]
    \item $\hat{u}_p(\cdot) \in \mathbb{U}_p(t_{\text{i}})$ (Planning system control constraint satisfaction)
    \item $\hat{\xi}_h(t) = \xi_h(t;p,t,\hat{u}_p(\cdot))$ (Planning system dynamics satisfaction)
    \item For all $t \in [t_{\text{i}},t_{\text{f}}]$, $\hat{\xi}_h(t) \in  \mathcal{C}_p(t)$ (Planning system state constraint satisfaction)
    \item There exists $T_r \in [t_{\text{i}},t_{\text{f}}]$ such that $\hat{\xi}_h(T_r) \in \mathcal{G}_p(t)$ (Planning system goal reaching) (optional)
\end{enumerate}

\begin{remark} \label{remark:conditions}
\ssreview{We consider C4 a \textit{soft constraint} in practice, since requiring there is some time $T_r$ so that the planning system is guaranteed to hit the planning system goal region is restrictive when the goal $\mathcal{G}_p$ is far from the system and the interval $[t_i,t_f]$ is short. However, it is} important for connecting Alg. \ref{alg:finite-interval-method} to the analysis in \ssminor{Sec.} \ref{sec:planning-analysis}. 
\end{remark}

\begin{remark} \label{remark:plan-generator}
    \ssminor{Our framework supports arbitrary} methods for generating the planned trajectory in line 5, \ssupdated{as} long as it can satisfy C1-C4. One \ssminor{possibility} is MPC, which will be demonstrated in \ssminor{Sec.} \ref{sec:case-study-planning}.
\end{remark}

The trajectory $\hat{\xi}_h(t), t \in [t_{\text{i}},t_{\text{f}}]$ is then returned as the output from the planning procedure. Note that \ssreview{for} static constraints, the optimal tracking controller  \eqref{eqn:optimal-tracking-controller} can already be applied to follow the planned trajectory to achieve constraint satisfaction and goal-reaching.


\subsection{Theoretical Guarantees for Tracking Planned Trajectories} \label{sec:planning-analysis}

In this section, we theoretically justify combining the optimal tracking controller with the trajectory planning method in Alg. \ref{alg:planning-bounded-interval} by providing Thm. \ref{theorem:planning-theorem}, which forms the basis of our analysis for trajectory replanning in later sections. Before doing so, let us define the \textit{optimal tracking system trajectory} $\xi_f^*:[t,T_{\text{off}}] \rightarrow \mathcal{S}$ and the \textit{worst-case planning system trajectory} $\xi_h^*:[t,T_{\text{off}}] \rightarrow \mathcal{P}$ in \eqref{eqn:optimal-tracking-system-trajectory} and \eqref{eqn:worst-planning-system-trajectory} respectively:
\begin{align}
    & \xi_f^*(t';s,p,t) := \xi_f(t'; s, t, u_s^*(\cdot; Ls-Mp, t), \label{eqn:optimal-tracking-system-trajectory} \\
    & \qquad d^*(\cdot; Ls-Mp, t)), \nonumber \\
    &\xi_h^*(t';p,s,t) := \xi_h(t'; p, t, u_p^*(\cdot; Ls-Mp, t)), \label{eqn:worst-planning-system-trajectory}
\end{align}
They correspond to the tracking system \eqref{eqn:time-varying-tracking-model} and planning system \eqref{eqn:time-varying-planning-model} trajectories starting from states $s$ and $p$ respectively, under the optimal tracking system controls \eqref{eqn:optimal-tracking-controls}, and worst-case planning system controls \eqref{eqn:worst-planning-controls} and disturbances \eqref{eqn:worst-disturbances}.


We now provide Lem. \ref{lemma:obstacle-avoidance}. It says \ssupdatedd{that} if the value level $c$ is sufficiently large and the planning system is initialized \ssreview{in} the planning system sublevel set $\mathcal{T}_p$, \ssupdatedd{then if} the worst-case planning system trajectory \eqref{eqn:worst-planning-system-trajectory} is \ssreview{in} $\mathcal{F}_B$, the optimal tracking system trajectory is \ssreview{in} $B$.

\begin{lemma} \label{lemma:obstacle-avoidance}

Consider $s \in \mathcal{S}$, $p \in \mathcal{P}$, $t \in [0,T]$, $c \in \mathbb{R}$ and $B \subseteq \mathcal{S}$. Suppose $c \geq \underbar{V}(s,t)$, and $p \in \mathcal{T}_p(s,t;c)$. Then, we have
\begin{equation}
    \xi_h^*(t';p,t) \in \mathcal{F}_B(t';c) \implies \xi_f^*(t';s,t) \in B,
\end{equation}
for all $t' \in [t,T]$.
\end{lemma}
Thm. \ref{theorem:planning-theorem} follows as an application of Lem. \ref{lemma:obstacle-avoidance} and says that under the same conditions as Lem. \ref{lemma:obstacle-avoidance}, if the worst-case planning system trajectory satisfies the planning system constraints $\mathcal{F}_{\mathcal{C}}(t;c)$, the optimal tracking system trajectory achieves constraint satisfaction in $\mathcal{C}$. Moreover, if the worst-case planning system trajectory satisfies the planning system goal $\mathcal{F}_{\mathcal{G}}(t;c)$, the optimal tracking system trajectory \ssupdatedd{is in} the goal $\mathcal{G}$.
\begin{theorem} \label{theorem:planning-theorem}
Consider tracking system state $s\in \mathcal{S}$ and planning system state $p \in \mathcal{P}$ at planning start time $t_{\text{i}} \in [0,T_{\text{off}}]$. Let $t_{\text{f}} \in [t_{\text{i}},T_{\text{off}}]$ be the planning end time, and let $c \in \mathbb{R}$ denote the value level. Suppose $c$ \ssminor{and $p$ satisfy}
\begin{equation}
    c \geq \underbar{V}(s,t_i) \quad  \ssminor{\text{and} \quad p \in \mathcal{T}_p(s,t_i;c).} \label{eqn:value-function-level-condition/initial-planning-system-condition}
\end{equation}
Then for any $t \in [t_{\text{i}},t_{\text{f}}]$, \ssminor{the following hold:}
\begin{align}
    &{\xi}_h^*(t;p,t_i) \in \mathcal{F}_{\mathcal{C}}(t;c) \implies {\xi}_f^*(t;p,t_i) \in \mathcal{C}, \label{eqn:planning-obstacle-avoidance} \\
    &\xi_h^*(t;p,t_i) \in \mathcal{F}_{\mathcal{G}}(t;c) \implies {\xi}_f^*(t;s,t_i)\in \mathcal{G}. \label{eqn:planning-goal-reaching}
\end{align}
\end{theorem}
\begin{remark} \label{remark:planning-arbitrary}
    Thm. \ref{theorem:planning-theorem} provides sufficient conditions for constraint satisfaction and goal reaching for the tracking system \eqref{eqn:time-varying-tracking-model} over $[t_{\text{i}},t_{\text{f}}]$, supposing that the optimal tracking system controls $u_s^*$ \eqref{eqn:optimal-tracking-controls} and worst-case disturbances $d^*$ \eqref{eqn:worst-disturbances} are applied to the tracking system \eqref{eqn:time-varying-tracking-model}, and the worst-case planning system controls $u_p^*$ \eqref{eqn:worst-planning-controls} are applied to the planning system \eqref{eqn:time-varying-planning-model}. Since Thm. \ref{theorem:planning-theorem} relies on Prop. \ref{prop:remain-in-teb}, the result should still hold when these are replaced by the optimal tracking controller \eqref{eqn:optimal-tracking-controller}, and arbitrary disturbances and planning system controls in $\mathcal{D}$ and $\mathcal{U}_p$, for the same reasons pointed out in Remark \ref{remark:practical-conditions}.
\end{remark}
We now discuss the connections between Thm. \ref{theorem:planning-theorem} and the scheme where the optimal tracking controller \eqref{eqn:optimal-tracking-controller} is applied to track trajectories planned via Alg. \ref{alg:planning-bounded-interval}. Condition C2 ensures that the planning system evolves according to \eqref{eqn:time-varying-planning-model}, and C1 ensures the planning system controls take values in $\mathcal{U}_p$, so supposing that the disturbances take values in $\mathcal{D}$, the guarantees in Thm. \ref{theorem:planning-theorem} are still applicable to the tracking system controlled by the proposed scheme via Remark \ref{remark:planning-arbitrary} --- we only need to check the conditions in Thm. \ref{theorem:planning-theorem} hold after replacing the worst-case planned trajectory $\xi_h^*$ with $\hat{\xi}_h$. In Alg. \ref{alg:planning-bounded-interval}, the value level $c$ \ssminor{and the planning system state $p$ satisfy} \eqref{eqn:value-function-level-condition/initial-planning-system-condition}, and the planned trajectory $\hat{\xi}_h$ satisfies C3 and (optionally) C4. Therefore, \eqref{eqn:planning-obstacle-avoidance} and \eqref{eqn:planning-goal-reaching} can be applied to conclude that constraint satisfaction and (optionally) goal-reaching are achieved over the interval $[t_{\text{i}},t_{\text{f}}]$.

Although safety is guaranteed when $\mathcal{C}$ is static, the current procedure does not handle constraints updated online. This limitation motivates trajectory replanning in \ssminor{Sec.} \ref{sec:replanning}.

\section{Trajectory Replanning and Tracking} \label{sec:replanning}

Although trajectory planning using Alg. \ref{alg:planning-bounded-interval} from \ssminor{Sec.} \ref{sec:planning} can be combined with tracking using \eqref{eqn:optimal-tracking-controller} to achieve static constraint satisfaction and goal-reaching, our problem setup in \ssreview{\ref{sec:problem-method-overview}} involves constraints \ssupdatedd{that} are updated online in real time. We clarify exactly what we mean by online constraint updates in \ssminor{Sec.} \ref{sec:online-constraints}. Subsequently, we propose a method that handles online constraint updates by replanning trajectories, over a task horizon $T_{\text{run}} \leq T_{\text{off}}$ in \ssminor{Sec.} \ref{sec:replanning-finite}. Moreover, our method allows for flexible reinitialization of the planning system, which can be exploited for improved performance. The theoretical properties of our method are subsequently \ssupdatedd{analyzed} in \ssminor{Sec.} \ref{sec:replanning-analysis} with respect to our objectives of online constraint satisfaction and goal-reaching. 

\subsection{The Online Constraint Satisfaction Problem} \label{sec:online-constraints}
Throughout this section, we consider constraints that are updated online in real time, denoted as $\mathcal{C}(t) \subseteq \mathcal{S}$ for $t \in [0,T_{\text{run}}]$. Since $\mathcal{C}(t)$ represents online constraints, at time $t \in [0,T_{\text{run}}]$, $\mathcal{C}(t')$ is unknown for $t' > t$. Moreover, it is assumed that $\mathcal{C}(t)$ changes at a finite number of increasing time instances $\{t^{\text{update}}_i\}_{i = 0}^{N_{\text{update}}} \subseteq [0,T_{\text{run}}]$ with $t^{\text{update}}_0=0$ and $t^{\text{update}}_{N_{\text{update}}}=T_{\text{run}}$, where $N_{\text{update}}$ is the number of times the constraints are updated. Specifically, for all $i \in \{0,\hdots,N_{\text{update}}-1 \}$ and $t_1,t_2 \in [t_i^{\text{update}},t_{i+1}^{\text{update}})$, we assume that
\begin{equation}
    \mathcal{C}(t_1) = \mathcal{C}(t_2). \label{eqn:assume-obstacles-updated-finite}
\end{equation}

\subsection{Method for Replanning and Tracking} \label{sec:replanning-finite}

Our online trajectory replanning and tracking method for time-varying systems subject to disturbances $d \in \mathcal{D}$ over a bounded time interval $t \in [0,T_{\text{run}}]$ is summarised in Alg. \ref{alg:finite-interval-method}.

\begin{algorithm}
\caption{Safe Trajectory Replanning and Tracking for Time-Varying Systems}
\begin{algorithmic}[1]
    \STATE \textbf{Input:} $L \in \mathbb{R}^{n_r \times n_s}$, $M \in \mathbb{R}^{n_r \times n_p}$, $T_{\text{off}} \leq T$, $f(t,s,u_s,d) \text{ for } t \in [0,T_{\text{off}}]$, $h(p,d)$, $\mathcal{G} \subseteq \mathcal{S}$, $T_{\text{run}} \leq T_{\text{off}}$
    \STATE Compute $V(r,t), t \in [0,T_{\text{off}}]$ and $u_s^*(r,t), t \in [0,T_{\text{off}}]$ using OfflineStage$(L,M,C,f,h)$ from Alg. \ref{alg:offline-stage}
    \STATE $k \leftarrow -1$
\FOR{$t \in [0,T_{\text{run}})$}
    \STATE Measure current tracking system state $s$ from \eqref{eqn:time-varying-tracking-model}
    \IF{$s \in \mathcal{G}$}
        \STATE Return (goal reached).
    \ENDIF
    \STATE Update constraints $\mathcal{C}_{\text{curr}} \leftarrow \mathcal{C}(t)$
    \IF{($\mathcal{C}_{\text{curr}}$ has changed) OR (decide to replan at time $t$ 
    
    (see Remark \ref{remark:choose-replan}))}
        \STATE $k \leftarrow k + 1$
        \STATE $(t_k,s_k) \leftarrow (t,s)$
        \STATE Choose $c_k \geq \underbar{V}(s_k,t_k)$ from \eqref{eqn:minimum-level-value-function} (see Remark \ref{remark:replan-level})
        \STATE Choose $p_k \in \mathcal{T}_p(s_k,t_k;c_k)$ from \eqref{eqn:planning-tube} (See Remark 
        
        \ref{remark:replan-level})
        \STATE $\hat{\xi}_{h,k}(t'), t' \in [t_k,T_{\text{run}}] \leftarrow \text{PlanTrajectoryOverInterval}$
        
        $  ( V, h, t_k, s_k, T_{\text{run}}, \mathcal{C}_{\text{curr}}, \mathcal{G}, c_k,p_k)$ from Alg. \ref{alg:planning-bounded-interval}
    \ENDIF
    \STATE $\hat{\xi}_h^R(t)\leftarrow \hat{\xi}_{h,k}(t)$
    \STATE Compute $u_s \leftarrow u_s^*(Ls-M\hat{\xi}_h^R(t),t)$ following \eqref{eqn:optimal-tracking-controller}
    \STATE Apply $u_s$ as control input to tracking system \eqref{eqn:time-varying-tracking-model}
    \ENDFOR
\end{algorithmic}    \label{alg:finite-interval-method}
\end{algorithm}

The inputs to the procedure include the matrices $L \in \mathbb{R}^{n_r \times n_s}$ and $M \in \mathbb{R}^{n_r \times n_p}$ required for formulating the relative model \eqref{eqn:time-varying-relative-model}, the tracking system model \eqref{eqn:time-varying-tracking-model}, the planning system model \eqref{eqn:time-varying-planning-model}, the goal $\mathcal{G} \subseteq \mathcal{S}$ and the task horizon $T_{\text{run}}$. Firstly, during the offline stage of the algorithm, we compute the value function $V(r,t)$ and the optimal tracking controller $u_s^*(r,t)$ over $t \in [0,T_{\text{off}}]$ via OfflineStage (Alg. \ref{alg:offline-stage}). Subsequent lines are then implemented online in real time. In particular, for each time over the task horizon $t \in [0,T_{\text{off}}]$, the current state $s$ of the tracking system is measured, and if $s$ is inside $\mathcal{G}$, then the algorithm stops running. The current state constraint \ssupdatedd{set} considered within the algorithm --- $\mathcal{C}_{\text{curr}}$ --- is \ssupdatedd{selected} based on the most recent measurements of the updated constraints $\mathcal{C}(t)$. If  $\mathcal{C}_{\text{curr}}$ has changed or we choose to replan at the current time $t$, then $k$ is incremented, we store the current time and tracking system state as $t_k$ and $s_k$ respectively, choose a value level $c_k$ satisfying $c_k \geq \underbar{V}(s_k,t_k)$, choose an initial planning system state $p_k$ satisfying $p_k \in \mathcal{T}_p(s_k,t_k;c_k)$, and then apply the trajectory replanning procedure \textit{PlanOverBoundedInterval} from Alg. \ref{alg:planning-bounded-interval} to solve for a trajectory $\hat{\xi}_{h,k}(t')$ over $t' \in [t_k,T_{\text{run}}]$. Outside of the trajectory replanning procedure, the current state of the planning system $\hat{\xi}_h^R(t)$ is set equal to $\hat{\xi}_{h,k}(t)$, and the optimal tracking controller \eqref{eqn:optimal-tracking-controller} is applied to the tracking system \eqref{eqn:time-varying-tracking-model} based on the current relative state $Ls_k-M\hat{\xi}_h^R(t)$ and current time $t$. The online procedure repeats until $T_{\text{run}}$. \ssreview{The online phase of Alg.~\ref{alg:finite-interval-method} is visualized as a block diagram in Fig.~\ref{fig:overview-replanning-tracking}.}

\begin{figure}[h]
    \centering
    \includegraphics[width=0.49\textwidth]{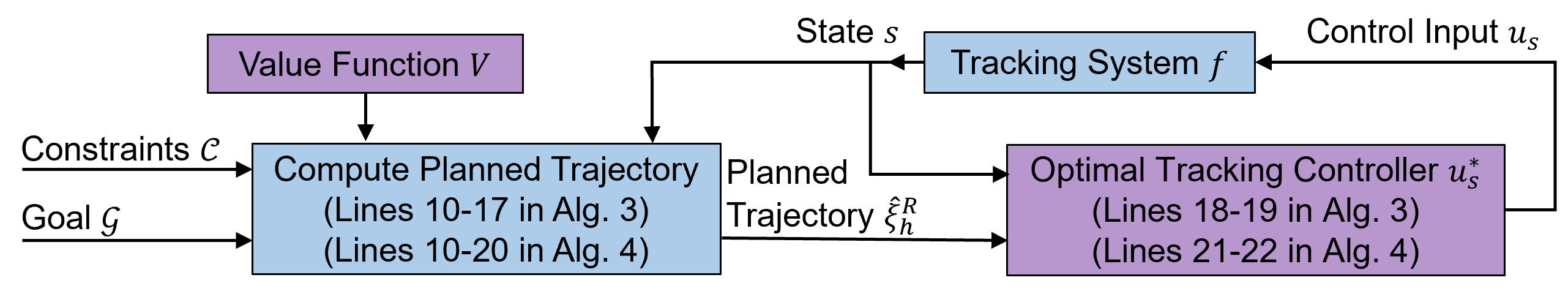}
    \caption{\ssreview{Block diagram \ssupdated{for online phase of Algs.~\ref{alg:finite-interval-method} and \ref{alg:periodic-method}}. Objects computed offline are colored purple.}}
    \label{fig:overview-replanning-tracking}
\end{figure}

The most expensive part of Alg. \ref{alg:finite-interval-method} is line 2 when OfflineStage is called. The online stage (lines 4-20) can be \ssreview{executed} in real time, \ssupdated{as} long as a computationally efficient real-time planning procedure is selected in line 15 (recall Remark \ref{remark:plan-generator}). 
\ssreview{This is possible by taking advantage of the framework's support for simple planning models.}

Note that an alternative strategy for choosing controls in line 18 is possible, where an arbitrary controller is used when the relative system $Ls-M\hat{\xi}_h^R(t)$ is within the interior of the sublevel set of $V(r,t)$ corresponding to $c_k$, and \eqref{eqn:optimal-tracking-controller} is applied when the relative system is close to the border. We refer the reader to \cite{Chen2021FaSTrack:Tracking} for further details.

\begin{remark} \label{remark:choose-replan}
Without loss of generality, trajectory planning must always be enforced at time $t=0$ and when constraints are updated. Beyond this, our method can flexibly support multiple different conditions in line 10. For example, replanning may be performed after every time interval of some fixed length, as is common in receding horizon control.
\end{remark}

\begin{remark} \label{remark:replan-level}
We do not prescribe a fixed way to set $c_k$ and $p_k$ in Alg. \ref{alg:finite-interval-method}, except they must satisfy $c_k \geq \underbar{V}(s_k,t_k)$ and $p_k \in \mathcal{T}_p(s_k,t_k;c_k)$. This is unlike FaSTrack \cite{Chen2021FaSTrack:Tracking}, which reinitializes the planning system state as equal to the state at time $t$ from the previous planned trajectory in line 15. In other words, our planning system \ssupdatedd{can} ``teleport'', which can be exploited for adaptive performance in the constraint satisfaction and goal-reaching problem. Specifically, choosing larger $c_k$ allows $p_k$ to be initialized further away from $s_k$, which encourages the optimal tracking controller \eqref{eqn:optimal-tracking-controller} to apply controls to catch up. On the other hand, smaller $c_k$ will produce less restrictive planning system constraints $\mathcal{C}_p(t)$, which may be required for navigation through tight environments. This tradeoff will be demonstrated inside simulations in \ssminor{Sec.} \ref{sec:case-study-replanning}. 
\end{remark}

\subsection{Theoretical Guarantees for Replanning and Tracking} \label{sec:replanning-analysis}

In this section, we theoretically justify the usage of Alg. \ref{alg:finite-interval-method} to compute controls for the tracking system \eqref{eqn:time-varying-tracking-model} \ssreview{via} Cor. \ref{cor:replanning-bounded-interval}. Before \ssreview{doing so}, we introduce mathematical objects required for our analysis.

Let $K \in \mathbb{N}$ denote the number of times trajectory replanning occurs, and suppose $T_{\text{run}} \leq T_{\text{off}} $. Consider a sequence of planning system states $\{ p_k \}_{k=0}^{K-1} \subseteq \mathcal{P}$, and a sequence of increasing time instances $\{ t_k \}_{k=0}^K \subseteq [0,T_{\text{off}}]$ where replanning occurs, satisfying $t_0 = 0$ and $t_K = T_{\text{run}}$. 
Assume $\{t_i^{\text{update}}\}_{i=0}^{N_{\text{update}}} \subseteq \{ t_k \}_{k=0}^K$ holds, \ssreview{so} replanning occurs whenever constraints are updated. 
We define the \textit{replanned trajectory} $\xi_h^R:[0,T_{\text{run}}) \rightarrow \mathcal{P}$ as the concatenation of the worst-case planning system trajectories \eqref{eqn:worst-planning-system-trajectory} \ssreview{on each sub-interval} $[t_k, t_{k+1})$ for $0 \leq k \leq K-1$, where each trajectory is determined by the planning system state $p_k$ and tracking system state $s_k$ at time $t_k$:
\begin{align}
    &\xi_h^R(t) :=
    \xi_h^*(t; p_k, s_k, t_k) \quad t \in [t_k, t_{k+1}), \label{eqn:replanned-trajectory}
\end{align}
for $0 \leq k \leq K-1$. Here, $s_k$ is recursively defined based on the optimal tracking system trajectory \eqref{eqn:optimal-tracking-system-trajectory} initialized at tracking system state $s_{k-1}$ and $p_{k-1}$ at time $t_{k-1}$:
\begin{equation}
    s_k := \xi_f^*(t_k ; s_{k-1}, p_{k-1}, t_{k-1}),
\end{equation}
and $s_0$ is the state of the tracking system \eqref{eqn:time-varying-tracking-model} at time $0$.

Next, define the \textit{replanned tracking system controls} $u_s^R:[0,T_{\text{run}}) \rightarrow \mathcal{U}_s$ as the concatenation of the optimal tracking system controls \eqref{eqn:optimal-tracking-controls} over each sub-interval between replanning time instances:
\begin{align}
    &u_s^R(t) := 
    u_s^*(t; Ls_k - Mp_k, t_k) \quad t \in [t_k, t_{k+1}), \label{eqn:replanned-tracking-controls}
\end{align}
for $0 \leq k \leq K-1$. Moreover, define the \textit{replanned disturbances} $d^R:[0,T_{\text{run}})\rightarrow\mathcal{D}$ as the concatenation of the worst-case disturbances \eqref{eqn:worst-disturbances} over each sub-interval:
\begin{align}
    &d^R(t) := 
    d^*(t; Ls_k - Mp_k, t_k) \quad t \in [t_k, t_{k+1}).\label{eqn:replanned-disturbances}
\end{align}
for $0 \leq k \leq K-1$. \ssreview{Next}, the \textit{replanned tracking system trajectory} is the trajectory of the tracking system \ssreview{starting from} $s_0$ at time $0$ under the replanned tracking system controls \eqref{eqn:replanned-tracking-controls} and disturbances \eqref{eqn:replanned-disturbances}:
\begin{align}
    &\xi_f^R(t) := \xi_f(t; s_0, 0, u_s^R(\cdot), d^R(\cdot)), t \in [0,T_{\text{run}}). \label{eqn:replanned-tracking-system-trajectory}
\end{align}
By definition, $\xi_h^R(t)$,
$u_s^R(t)$, $d^R(t)$ over $t \in [0,T_{\text{off}})$ are uniquely determined by $s_0$, $\{p_k \}_{k=0}^{k'(t)}$ and $\{t_k \}_{k=0}^{k'(t)}$, where $k'(t) = \max \{ k \in \mathbb{N}_0 : t_k \leq t \}$.

With these mathematical objects defined, we are now ready to provide Cor. \ref{cor:replanning-bounded-interval}. It says that for each replanning step $k \in \{0,\hdots,K-1\}$ and associated time interval $[t_k,t_{k+1})$, \ssupdated{as} long as the corresponding value level $c_k$ is sufficiently large and the initial planning system state $p_k$ on this interval is inside the planning system sublevel set $\mathcal{T}_p(\xi_f^R(t_k),t_k;c_k)$ \eqref{eqn:minimum-level-value-function-replanning/viable-planning-states-replanning}, then at any time in this interval that the replanned trajectory $\xi_h^R(t)$ satisfies the planning system constraints $\mathcal{F}_{\mathcal{C}(t)}(t;c_k)$ (goal $\mathcal{F}_{\mathcal{G}}(t;c_k)$), the replanned tracking system trajectory $\xi_f^R(t)$ satisfies the constraints $\mathcal{C}(t)$ \eqref{eqn:replanning-obstacle-avoidance} (goal $\mathcal{G}$ \eqref{eqn:replanning-goal-reaching}).

\begin{corollary} \label{cor:replanning-bounded-interval}
Let $\{ c_k \}_{k=0}^{K-1} \subseteq \mathbb{R}$ be a sequence of value levels. Suppose for all $k \in \{ 0, \hdots, K-1 \}$,
\begin{equation}
    c_k \geq \underbar{V}(\xi_f^R(t_k),t_k) \quad \ssminor{\text{and}} \quad p_k \in \mathcal{T}_p(\xi_f^R(t_k),t_k;c_k). \label{eqn:minimum-level-value-function-replanning/viable-planning-states-replanning}
\end{equation}
Then, for any $k \in \{0, \hdots, K-1\}$ and any $t \in [t_{k},t_{k+1})$\ssminor{, we have:}
\begin{align}
    &\xi_h^R(t) \in \mathcal{F}_{\mathcal{C}(t_k)}(t';c_k) \implies \xi_f^R(t) \in \mathcal{C}(t),  \label{eqn:replanning-obstacle-avoidance} \\
    &\xi_h^R(t) \in \mathcal{G}_p(t;c_k) \implies \xi_f^R(t) \in \mathcal{G}. \label{eqn:replanning-goal-reaching}
\end{align}
\end{corollary}
\begin{remark} \label{remark:replanning-arbitrary}
    Cor. \ref{cor:replanning-bounded-interval} provides sufficient conditions for collision avoidance and/or goal-reaching over the task horizon $[0,T_{\text{run}})$ supposing that within each sub-interval $[t_k,t_{k+1})$ for $k \in \{ 0, \hdots, K-1 \}$, the optimal tracking system controls $u_s^*$ \eqref{eqn:optimal-tracking-controls} and worst-case disturbances \eqref{eqn:worst-disturbances} are applied to \eqref{eqn:time-varying-tracking-model} and the worst-case planning system controls \eqref{eqn:worst-planning-controls} are applied to \eqref{eqn:time-varying-planning-model}, after reinitializing the planning system state at $p_k$. The proof naturally follows by applying Thm. \ref{theorem:planning-theorem} in each sub-interval of the task horizon, and therefore the guarantees should also hold when the optimal tracking controller \eqref{eqn:optimal-tracking-controller} and arbitrary planning system controls and disturbances taking values in $\mathcal{U}_p$ and $\mathcal{D}$ are applied within each sub-interval, for the same reasons pointed out in Remark \ref{remark:planning-arbitrary}. 
\end{remark}

We now discuss the connections between Alg. \ref{alg:finite-interval-method} and Cor. \ref{cor:replanning-bounded-interval}. Note that every time replanning occurs in Alg. \ref{alg:finite-interval-method}: 1) the planning system controls for each replanned trajectory take values in $\mathcal{U}_p$ since C1 must be satisfied when planning the trajectory in line 15; 2) the tracking system controls are generated by the optimal tracking controller \eqref{eqn:optimal-tracking-controller} in line 18; 3) the disturbances take values in $\mathcal{D}$ by assumption. Thus, the results in Cor. \ref{cor:replanning-bounded-interval} \ssupdatedd{apply} to the tracking system controlled by Alg. \ref{alg:finite-interval-method}, as discussed in Remark \ref{remark:replanning-arbitrary}. To see that collision avoidance or goal-reaching is achieved, we only need to check that the conditions in Cor. \ref{cor:replanning-bounded-interval} are satisfied after replacing $\xi_h^R(t)$ and $\xi_f^R(t_k)$ with $\hat{\xi}_h^R(t)$ and $s_k$ from Alg. \ref{alg:finite-interval-method}. This can be verified by comparing \ssupdated{lines} 13 \ssminor{and 14} in Alg. \ref{alg:finite-interval-method} to \ssminor{\eqref{eqn:minimum-level-value-function-replanning/viable-planning-states-replanning}}, and $\hat{\xi}_f^R(t)=\hat{\xi}_{h,k}(t)$ (which satisfies C3 and C4 via line 15) to \eqref{eqn:replanning-obstacle-avoidance} and \eqref{eqn:replanning-goal-reaching}.
Therefore, \eqref{eqn:replanning-obstacle-avoidance} and (optionally) \eqref{eqn:replanning-goal-reaching} can be applied to find that constraint satisfaction and goal-reaching are achieved within the corresponding sub-interval up to the next time $t_{k+1}$ that replanning occurs. Since replanning is repeated throughout the task horizon, we can conclude that Alg. \ref{alg:finite-interval-method} ensures the tracking system \eqref{eqn:time-varying-tracking-model} achieves constraint satisfaction and (optionally) goal reaching over the entire task horizon.

\begin{remark} \label{remark:tube}
    \ssreview{The argument in Cor.~\ref{cor:replanning-bounded-interval} is similar to the idea of sequentially composing tubes (or funnels) in \cite{Majumdar2017FunnelPlanning:}, but there are some differences. Recall that a tube is a family of subsets of the state space $\mathcal{S}$ parameterized over time $t$ which the tracking system is guaranteed to remain inside. For each sub-interval $[t_k,t_{k+1})$ where $0 \leq k \leq K-1$, there is an implicit tube $\mathcal{T}_s(t;\xi_h^R\ssupdated{(\cdot)},c_k):=\{ s \in \mathcal{S} : V(Ls-M\xi_h^R(t), t) \leq c_k \}$ associated with the planned trajectory $\xi_h^R$ and value level $c_k$. Condition \eqref{eqn:minimum-level-value-function-replanning/viable-planning-states-replanning} \ssupdatedd{ensures} that the `inlet' of each tube contains the specific tracking system at the corresponding replanning time, as \ssupdatedd{illustrated} in Fig.~\ref{fig:visualise-tubes} a). This is unlike \cite{Majumdar2017FunnelPlanning:}, which requires that the `outlet' of the previous tube is a subset of the `inlet' of the current tube, as illustrated in Fig.~\ref{fig:visualise-tubes} b).}
\end{remark}

\begin{figure} 
    \centering
    \includegraphics[width=0.49\textwidth]{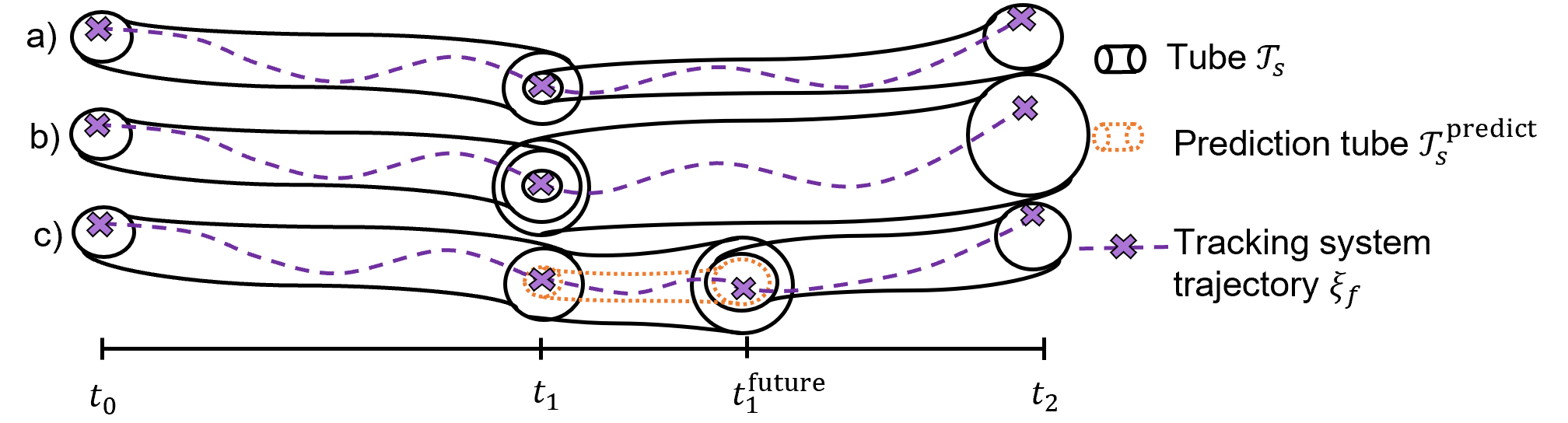}
    \caption{\ssreview{Illustration of tubes containing the tracking system trajectory. 
    Case a) represents the method in this paper where the inlet just needs to contain the tracking system state. 
    Case b) is analogous to \cite{Majumdar2017FunnelPlanning:} where the outlet of the previous tube is a subset of the inlet of the current tube.
    Case c) represents a modification to handle non-instantaneous trajectory planning.}}
    \label{fig:visualise-tubes}
\end{figure}

\begin{remark} \label{remark:non-instant-planning}
    \ssreview{Our strategy and analysis \ssupdatedd{do} not account for the time required to run the online part of Alg.~\ref{alg:finite-interval-method}, which is subject to hardware-related constraints. Essentially, we assume all lines in Alg.~\ref{alg:finite-interval-method} resolve instantaneously, which is not true, even when real-time planning methods exploiting simple planning models are chosen.
    A more realistic assumption is the online computation time is concentrated in line~15 of Alg.~\ref{alg:finite-interval-method} (during planning), since it dominates computation in practice.
    Although computation time can be drastically reduced using simple planning system models (e.g. single integrator), it would still not be instantaneous.
    However, assuming we know an upper bound $T_{\textnormal{compute}} \geq 0$ on the computation time for planning, we can modify Alg.~\ref{alg:finite-interval-method} to account for it. This is done by turning line~15 into a parallel process that generates a new planned trajectory from a \textit{future} time $t^{\textnormal{future}}_k$ greater than the current time $t_k$ plus $T_{\textnormal{compute}}$, while simultaneously tracking the \ssupdatedd{previously} planned trajectory with $u_s^*$. 
Note that the new planned trajectory cannot be initialized based on the `exact' tracking system state at $t^{\textnormal{future}}_k$ \ssupdatedd{analogously} to line~14 since this state is unknown. However, we can instead define a `prediction tube' $\mathcal{T}_s^{\textnormal{predict}}(t;s_k,\xi_h^R\ssupdated{(\cdot)}):= \mathcal{T}_s(t;\xi_h^R\ssupdated{(\cdot)},V(Ls_k - M\xi_h^R(t_k),t_k))$  guaranteed to contain the tracking system over $t \in [t_k,t^{\textnormal{future}}_k]$, and ensure that the inlet of the next tube contains its outlet, as \ssupdatedd{visualized} in Fig.~\ref{fig:visualise-tubes} c).
This is implemented by initializing the next planned trajectory at $t^{\textnormal{future}}_k$ in the set $\bigcap_{s \in \mathcal{T}_s^{\textnormal{predict}}(t^{\textnormal{future}}_k;s_k,\xi_h^R\ssupdated{(\cdot)})}\mathcal{T}_p(s,t^{\textnormal{future}}_k;c_k)$ with $c_k \geq \underline{c}$, and tracking the trajectory with $u_s^*$ after $t^{\textnormal{future}}_k$. Despite outlining the modifications to account for $T_{\textnormal{compute}}$, we continue assuming instantaneous planning in the paper for simplicity.}
\end{remark}
\section{Trajectory Replanning and Tracking for Periodic Systems} \label{sec:replanning-periodic}

The trajectory replanning and tracking method in \ssminor{Sec.} \ref{sec:replanning} can guarantee collision avoidance and \ssupdatedd{goal-reaching} with constraints updated online in real time but is only applicable over a restricted horizon $T_{\text{run}} \leq T_{\text{off}}$. We overcome this limitation for the specific case of periodic systems (\ssminor{Sec.} \ref{sec:periodic-systems}), allowing for $T_{\text{run}} \geq 0$ generally. We propose a trajectory replanning and tracking method for periodic systems in \ssminor{Sec.} \ref{sec:replanning-method-periodic}. The theoretical properties of our method are subsequently \ssupdatedd{analyzed} in \ssminor{Sec.} \ref{sec:replanning-analysis-periodic} with respect to our objectives of online constraint satisfaction and goal-reaching. 

\subsection{Periodic Tracking Systems} \label{sec:periodic-systems}
Throughout this section, assume that the tracking system dynamics are periodic: 
\begin{align}
    f(t,s,u_s,d)= f(t+\tau,s,u_s,d), \ t \in [0,\infty), \label{eqn:periodic-tracking}
\end{align}
where $u_s \in \mathcal{U}_s$, $u_p \in \mathcal{U}_p$, $d \in \mathcal{D}$, and $\tau \in [0,\infty)$ is the \ssupdated{period}. Although this seems like a strong assumption, there are practical settings where it is reasonable, such as underwater robots subject to wave disturbances \cite{Fernandez2017ModelWaves}. Moreover, we assume $T_{\text{off}} > \tau$ to ensure all time-varying information is captured during offline precomputation.

The usefulness of periodic systems is that methods and guarantees applicable over time intervals within the offline computation horizon $T_{\text{off}}$, are also applicable to intervals outside of it, where the dynamics are equivalent since the system is periodic. Fundamentally, the reason for this is because the value function $V(r,t)$ computed explicitly over $t\in[k \tau, k \tau + T_{\text{off}}]$ for $k \in \mathbb{N}_0$ is equivalent to $V(r,t - k \tau)$ for $t \in [k \tau, k \tau + T_{\text{off}}]$, where $t-k\tau \in [0,T_{\text{off}}]$ holds. Thus, by computing $V(r,t)$ over $[0,T_{\text{off}}]$, we have essentially obtained the value function (and optimal controller) computed over any interval $t \in [k \tau, k \tau + T_{\text{off}}]$ for $k \in \mathbb{N}_0$ for free, without explicitly performing the computation. This idea underlies the  analysis in \ssminor{Sec.} \ref{sec:replanning-analysis-periodic}.

To relate arbitrary times to the offline computation interval, we introduce the \textit{earliest equivalent interval map}:
\begin{equation}
    \mathcal{M}_{\tau}([t_{a},t_{b}]) := [t_{a} - \floor{t_{a} / \tau} \tau, t_{b} - \floor{t_{a} / \tau} \tau]
\end{equation}
for $t_{b} \geq t_{a} \geq 0$, where $\floor{\cdot}$ is the floor function. It maps an arbitrary time interval to the earliest interval after zero where the system dynamics are equivalent. Note that when $t_b - t_a \leq T_{\text{off}}- \tau$, the earliest equivalent map is within the offline computation interval, i.e. $\mathcal{M}_{\tau}([t_a,t_b]) \subseteq [0,T_{\text{off}}]$. Moreover, let $\mathcal{M}_{\tau}^{\text{i}}(t_a,t_b)=\min(\mathcal{M}_{\tau}([t_{a},t_{b}]))$ and $\mathcal{M}_{\tau}^{\text{f}}(t_a,t_b)=\max(\mathcal{M}_{\tau}([t_{a},t_{b}]))$ denote the initial and final times in the earliest equivalent interval respectively.
These functions will be used throughout \ssminor{Sec.} \ref{sec:replanning-method-periodic} and \ref{sec:replanning-analysis-periodic}.

Just like \ssminor{Sec.} \ref{sec:replanning}, we assume the constraints $\mathcal{C}(t)$ for $t \in [0, T_{\text{run}})$ are updated online at time instances $\{t_i^{\text{updated}}\}_{i=0}^{N_{\text{update}}}$, satisfying \eqref{eqn:assume-obstacles-updated-finite}.

\subsection{Method for Replanning and Tracking for Periodic Systems} \label{sec:replanning-method-periodic}

Our method for real-time trajectory replanning and tracking for periodic systems subject to disturbances $d \in \mathcal{D}$ over a bounded time interval $t \in [0,T_{\text{run}}]$ is summarised in Alg. \ref{alg:periodic-method}. The logic in Alg. \ref{alg:periodic-method} \ssupdated{is similar to Alg.~\ref{alg:finite-interval-method} in the sense that they both involve repeated planning and tracking, as illustrated in Fig.~\ref{fig:overview-replanning-tracking}, but the concrete steps are quite different. Rather than describing every line, we instead bring attention to these differences.}

\begin{algorithm}
\caption{Safe Trajectory Replanning and Tracking for Periodic Systems}
\begin{algorithmic}[1]
    \STATE \textbf{Input:} $L \in \mathbb{R}^{n_r \times n_s}$, $M \in \mathbb{R}^{n_r \times n_p}$, $T_{\text{off}} > \tau$, $f(t,s,u_s,d) \text{ for } t \in [0,T_{\text{off}}]$, $h(p,d)$, $\mathcal{G} \subseteq \mathcal{S}$, $T_{\text{run}} \geq 0$, $T' \leq T_{\text{off}} - \tau$
    \STATE Compute $V(r,t), t \in [0,T_{\text{off}}]$ and $u_s^*(r,t), t \in [0,T_{\text{off}}]$ using OfflineStage$(L,M,C,f,h)$ from Alg. \ref{alg:offline-stage}
    \STATE $k \leftarrow -1$, $t_{-1} \leftarrow 0$
\FOR{$t \in [0,T_{\text{run}})$}
    \STATE Measure current tracking system state $s$ from \eqref{eqn:time-varying-tracking-model}
    \IF{$s \in \mathcal{G}$}
        \STATE Return (goal reached)
    \ENDIF
    \STATE Update constraints $\mathcal{C}_{\text{curr}} \leftarrow \mathcal{C}(t)$
    \IF{($C_{\text{curr}}$ has changed) OR (choose to replan at time $t$) 
    
    OR ($t - t_k \geq  T'$)}
        \STATE $k \leftarrow k + 1$
        \STATE $(t_{k},s_k) \leftarrow (t,s)$
        \STATE $(t_{\text{i}},t_{\text{f}}) \leftarrow (\mathcal{M}^{\text{i}}_{\tau}(t_k,t_k + T'), \mathcal{M}^{\text{f}}_{\tau}(t_k,t_k + T'))$
        \STATE Choose $c_k \geq \underbar{V}(s_k,t_{\text{i}})$ from \eqref{eqn:minimum-level-value-function}
        \STATE Choose $p_k \in \mathcal{T}_p(s_k,t_{\text{i}};c_k)$ from \eqref{eqn:planning-tube}
        \STATE $\hat{\xi}^{\mathcal{M}}_{h,k}(t'),t' \in [t_{\text{i}},t_{\text{f}}] \leftarrow \text{PlanTrajectoryOverInterval}$
        
        $ ( V, h, t_{\text{i}}, s_k, t_{\text{f}}, \mathcal{C}_{\text{curr}}, \mathcal{G}, c)$ from Alg. \ref{alg:planning-bounded-interval}
        \STATE $\hat{\xi}_{h,k}(t') \leftarrow \hat{\xi}^{\mathcal{M}}_{h,k}(\mathcal{M}^{\text{f}}_{\tau}(t_k,t')),t' \in [t_k,t_k + T']$
    \ENDIF
    \STATE $t_{\text{c}} \leftarrow \mathcal{M}^{\text{f}}_{\tau}([t_k,t])$
    \STATE $\hat{\xi}_h^R(t) \leftarrow \hat{\xi}_{h,k}(t)$
    \STATE Compute $u_s \leftarrow u_s^*(Ls-M\hat{\xi}_h^R(t), t_c)$ following \eqref{eqn:optimal-tracking-controller}
    \STATE Apply $u_s$ as control input to tracking system \eqref{eqn:time-varying-tracking-model}
    \ENDFOR
\end{algorithmic} \label{alg:periodic-method}
\end{algorithm}

Firstly, we allow for $T_{\text{run}} \geq 0$, enabling task horizons longer than the offline computation interval $T_{\text{off}}$. Moreover, there is a new input $T' < T_{\text{off}} - \tau$, which is the \textit{planning horizon} --- that is, the time horizon over which trajectory planning occurs, which ensures that the interval over which trajectories are planned can be successfully mapped back to the offline computation horizon. Subsequently, in line 10, we force replanning when $t-t_k \geq T'$. In line 13, $(t_{\text{i}},t_{\text{f}})$ are computed. They correspond to the initial and final time after mapping the interval $[t_k, t_k + T']$ into the offline computation horizon. The initial time $t_i$ is subsequently used to specify the minimum value level $\underbar{V}(s_k,t_{\text{i}})$ in line 14, and the initial planning system sublevel set $\mathcal{T}_p(s_k,t_{\text{i}};c_k)$ in line 15. \textit{PlanTrajectoryOverInterval} is applied to generate a trajectory $\hat{\xi}^{\mathcal{M}}_{h,k}(t')$ over the remapped interval $t' \in [t_{\text{i}},t_{\text{f}}]$ in line 16. The planned trajectory is then mapped back to the original interval, yielding $\hat{\xi}_{h,k}(t'), t' \in [t_k, t_k + T']$ in line 17.

The final differences are lines 19, 20, and 21. The current time $t$ is mapped back to the offline interval to obtain $t_c$ in line 19, the planned trajectory in the offline interval $\hat{\xi}^{\mathcal{M}}_{h,k}(t_c)$ is used to obtain the current state of the replanned trajectory $\hat{\xi}_h^R(t)$ in line 20, and these are subsequently used to compute the control in line 21.

All of the changes described between Alg. \ref{alg:finite-interval-method} and Alg. \ref{alg:periodic-method} address the need for objects related to the offline computation --- i.e. the value function $V$ \eqref{eqn:value-function} or optimal tracking controller $u_s^*$ \eqref{eqn:optimal-tracking-controller} ---  to have their times mapped back to the offline interval where the system dynamics are equivalent. Making this change allows us to overcome the limitations in \ssminor{Sec.} \ref{sec:replanning} and deal with the case where $T_{\text{run}} > T_{\text{off}}$ for periodic systems. Aside from this, Alg. \ref{alg:periodic-method} inherits the same properties as Alg. \ref{alg:finite-interval-method}, including the requirement for careful design of the replanning criteria in line 10, flexibility to choose $c_k$ and $p_k$ in lines 14 and 15, and the capability of switching between the optimal tracking controller and an arbitrary controller in line 21.

\subsection{\mbox{Theoretical Guarantees for Replanning and Tracking }in Periodic Systems} 
\label{sec:replanning-analysis-periodic}
In this section, we theoretically analyze the usage of Alg. \ref{alg:periodic-method} for controlling the system \eqref{eqn:time-varying-tracking-model} with periodic dynamics \eqref{eqn:periodic-tracking} by providing Cor. \ref{cor:replanning-periodic}. Our analysis progresses similarly to \ssminor{Sec.} \ref{sec:replanning-analysis}. Before describing this result, we introduce the required mathematical objects.

Let $K \in \mathbb{N}$ denote the number of times trajectory replanning occurs. Consider a sequence of planning system states $\{ p_k \}_{k=0}^{K-1} \subseteq \mathcal{P}$, and a sequence of increasing time instances $\{ t_k \}_{k=0}^K$ where replanning occurs, satisfying $t_0 = 0$, $t_K = T_{\text{run}}$. Assume $\{t_i^{\text{update}}\}_{i=0}^{N_{\text{update}}} \subseteq \{t_k\}_{k=0}^K$, such that replanning occurs whenever constraints are updated. We define the \textit{replanned trajectory} $\xi_h^R:[0,T_{\text{run}}) \rightarrow \mathcal{P}$ as the concatenation of trajectories defined on the sub-intervals $[t_k, t_{k+1})$ for $0 \leq k \leq K-1$, corresponding to the worst-case planned trajectory \eqref{eqn:worst-planning-system-trajectory} over the interval $[\mathcal{M}^{\text{i}}_{\tau}(t_k,t_{k+1}),\mathcal{M}^{\text{f}}_{\tau}(t_k,t_{k+1}))$ with initial planning system state $p_k$ and tracking system state $s_k$:
\begin{align}
    &\xi_h^R(t) :=
    \xi_h^*(\mathcal{M}^{\text{f}}_{\tau}(t_k,t); p_k, s_k, \mathcal{M}^{\text{i}}_{\tau}(t_k,t)) , \label{eqn:replanned-trajectory-periodic} \\
    &\quad t \in [t_k, t_{k+1}), \quad 0 \leq k \leq K-1. \nonumber
\end{align}
Here, $s_k$ is recursively defined as
\begin{equation}
    s_k := \xi_f^*(\mathcal{M}^{\text{f}}_{\tau}(t_{k-1},t_k) ; s_{k-1}, p_{k-1}, \mathcal{M}^{\text{i}}_{\tau}(t_{k-1},t_k)).
\end{equation}
and $s_0$ is the state of the tracking system at time $t=0$.

Next, define the \textit{replanned tracking system controls} $u_s^R:[0,T_{\text{run}}) \rightarrow \mathcal{U}_s$ as the concatenation of the optimal tracking system controls \eqref{eqn:optimal-tracking-controls}, computed over the remapped intervals in the offline computation horizon that correspond to each sub-interval between replanning time instances:
\begin{align}
    &u_s^R(t) := 
    u_s^*(\mathcal{M}^{\text{f}}_{\tau}(t_k,t); Ls_k - Mp_k, \mathcal{M}^{\text{i}}_{\tau}(t_k,t)), \\
    &\quad t \in [t_k, t_{k+1}), \ 0 \leq k \leq K-1.
\end{align}
Moreover, \ssupdatedd{we define the \textit{replanned disturbances} as} the concatenation of the worst-case disturbances \eqref{eqn:worst-disturbances} computed over the remapped intervals in the offline computation horizon that correspond to each sub-interval between replanning time instances:
\begin{align}
    &d^R(t) := 
    d^*(\mathcal{M}^{\text{f}}_{\tau}(t_k,t); Ls_k - Mp_k, \mathcal{M}^{\text{i}}_{\tau}(t_k,t)), \\
    &\quad t \in [t_k, t_{k+1}), \ 0 \leq k \leq K-1.
\end{align}

Finally, the \textit{replanned tracking system trajectory} is defined as the trajectory of the tracking system starting from state $s_0$ at time $0$ under the replanned tracking system controls \eqref{eqn:replanned-tracking-controls} and disturbances \eqref{eqn:replanned-disturbances}:
\begin{align}
    &\xi_f^R(t) := \xi_f(t; s_0, 0, u_s^R(\cdot), d^R(\cdot)), \quad t \in [0,T_{\text{run}}).
\end{align}
For all time $t \in [0,T_{\text{run}})$, $\xi_h^R(t)$,
$u_s^R(t)$, and $d^R(t)$, are uniquely determined by $s_0$, $\{p_k \}_{k=0}^{k'(t)}$ and $\{t_k \}_{k=0}^{k'(t)}$, where $k'(t) = \max \{ k \in \mathbb{N}_0 : t_k \leq t \}$.

With these mathematical objects defined, we are now ready to provide Cor. \ref{cor:replanning-periodic}, which says that, for each replanning step $k \in \{0,\hdots,K-1\}$ and associated time interval $[t_k,t_{k+1})$, \ssupdated{as} long as this time interval can be mapped back to the offline computation interval \eqref{eqn:successfully-mapped-offline-interval}, the corresponding value level $c_k$ is sufficiently large \eqref{eqn:minimum-level-value-function-replanning-periodic} and the initial planning system state $p_k$ on this interval is inside the planning system sublevel set $\mathcal{T}_p(\xi_f^R(t_k),\mathcal{M}^{\text{i}}_{\tau}(t_k,t_{k+1});c_k)$ \eqref{eqn:viable-planning-states-replanning-periodic}, then at any time in this interval that the replanned trajectory $\xi_h^R(t)$ is within the planning system constraints $\mathcal{F}_{\mathcal{C}(t_k)}(\mathcal{M}^{ \text{f} }_{\tau}(t_k,t);c_k)$ (goal $\mathcal{F}_{\mathcal{G}}(\mathcal{M}^{ \text{f} }_{\tau}(t_k,t);c_k)$), the replanned tracking system trajectory $\xi_f^R(t)$ satisfies the constraints $\mathcal{C}(t)$ \eqref{eqn:replanning-obstacle-avoidance-periodic} (goal $\mathcal{G}$ \eqref{eqn:replanning-goal-reaching-periodic}).

\begin{corollary} \label{cor:replanning-periodic}
Let $\{ c_k \}_{k=0}^{K-1} \subseteq \mathbb{R}$ be a sequence of value levels. Suppose for all $k \in \{ 0, \hdots, K-1 \}$\ssminor{, the following hold:
\begin{align}
    &t_{k+1} - t_k \leq T_{\text{off}} - \tau, \label{eqn:successfully-mapped-offline-interval} \\
    &c_k \geq \underbar{V}(\xi_f^R(t_k),\mathcal{M}^{\text{i}}_{\tau}(t_k,t_{k+1})), \label{eqn:minimum-level-value-function-replanning-periodic} \\
    & p_k \in \mathcal{T}_p(\xi_f^R(t_k),\mathcal{M}^{\text{i}}_{\tau}(t_k,t_{k+1});c_k). \label{eqn:viable-planning-states-replanning-periodic}
\end{align}}
Then, for any $k \in \{0, \hdots, K-1\}$ and any $t \in [t_{k},t_{k+1})$\ssminor{, we have:}
\begin{align}
    &\xi_h^R(t) \in \mathcal{F}_{\mathcal{C}(t_k)}(\mathcal{M}^{ \text{f} }_{\tau}(t_k,t);c_k) \implies \xi_f^R(t) \in \mathcal{C}(t),  \label{eqn:replanning-obstacle-avoidance-periodic} \\
    &\xi_h^R(t) \in \mathcal{G}_p(\mathcal{M}^{ \text{f} }_{\tau}(t_k,t);c_k) \implies \xi_f^R(t) \in \mathcal{G}. \label{eqn:replanning-goal-reaching-periodic}
\end{align}
\end{corollary}
\begin{remark} \label{remark:replanning-periodic-arbitrary}
    Similar to Cor. \ref{cor:replanning-bounded-interval}, Cor. \ref{cor:replanning-periodic} provides sufficient conditions for collision avoidance and/or goal-reaching over the task horizon supposing that within each sub-interval $[t_k,t_{k+1})$ for $k \in \{ 0, \hdots, K-1 \}$, optimal tracking system controls $u_s^*$ \eqref{eqn:optimal-tracking-controls} and worst-case disturbances $d^*$ \eqref{eqn:worst-disturbances} are applied to \eqref{eqn:time-varying-tracking-model}, and the worst-case planning system controls $u_p^*$ \eqref{eqn:worst-planning-controls} are applied to \eqref{eqn:time-varying-planning-model}, after accounting for the reinitialization of the planning system state at $p_k$. However, these optimal controls and disturbances are solved over remapped intervals inside the offline computation horizon, as opposed to the true sub-interval. Since the dynamics of the tracking system are equivalent across both due to the periodic nature of the system, this extra complexity allows us to provide guarantees for $T_{\text{run}} > T_{\text{off}}$, without requiring the solution of the HJ differential game beyond $T_{\text{off}}$. The proof of Cor. \ref{cor:replanning-periodic} naturally follows by applying Thm. \ref{theorem:planning-theorem} for each remapped interval corresponding to each sub-interval of the task horizon. Thus, for the same reasons pointed out in Remark \ref{remark:planning-arbitrary}, the results in Cor. \ref{cor:replanning-periodic} also hold when the optimal tracking controller \eqref{eqn:optimal-tracking-controller} for each remapped interval, and arbitrary planning system controls and disturbances taking values in $\mathcal{U}_p$ and $\mathcal{D}$ are applied in each sub-interval of the task horizon.
\end{remark}

The connections between Alg. \ref{alg:periodic-method} and the guarantees in Cor. \ref{cor:replanning-periodic} follow similarly to that of Alg. \ref{alg:finite-interval-method} and Cor. \ref{cor:replanning-bounded-interval}. In particular, the planning system controls for the replanned trajectories take values in $\mathcal{U}_p$, the tracking system controls are generated by the optimal tracking controller (line 21), and the disturbances take values in $\mathcal{D}$; therefore, referring to Remark \ref{remark:replanning-periodic-arbitrary}, the guarantees in Cor. \ref{cor:replanning-periodic} \ssupdatedd{apply} to the system controlled by Alg. \ref{alg:periodic-method} \ssupdated{as} long as the conditions in Cor. \ref{cor:replanning-periodic} hold after replacing $\xi_h^R(t)$ and $\xi_f^R(t_k)$ with $\hat{\xi}_h^R(t)$ and $s_k$ from Alg. \ref{alg:periodic-method}. 
This can be verified by comparing line 14 in Alg. \ref{alg:finite-interval-method} to \eqref{eqn:minimum-level-value-function-replanning-periodic}, line 15 to \eqref{eqn:viable-planning-states-replanning-periodic}, and $\hat{\xi}_f^R(t)=\hat{\xi}_{h,k}(t)$ (which satisfies C3 and C4 via line 16) to \eqref{eqn:replanning-obstacle-avoidance-periodic} and \eqref{eqn:replanning-goal-reaching-periodic}.
Following the same logic as \ssminor{Sec.} \ref{sec:replanning-analysis}, \eqref{eqn:replanning-obstacle-avoidance-periodic} and \eqref{eqn:replanning-goal-reaching-periodic} can therefore be applied to find that constraint satisfaction and (optionally) \ssupdatedd{goal-reaching} are achieved between replanning time instances, implying Alg. \ref{alg:periodic-method} achieves these guarantees over the entire task horizon.

\section{Example: AUV Subject to Wave Disturbances} \label{sec:auv-example}
\ssreview{This section explains how the methods in this paper can be implemented on a numerical example of an AUV, and showcases its performance.
Different ways of modeling an AUV subject to wave disturbances are provided in Sec.~\ref{sec:auv-model}, and in Sec.~\ref{sec:case-study-offline-computation} we describe how to interface these models with the OfflineStage (Alg.~\ref{alg:offline-stage}) component of our framework. In Sec.~\ref{sec:case-study-planning}, a concrete implementation of the planning procedure in Alg.~\ref{alg:planning-bounded-interval} is described, and its benefits are highlighted. Sec.~\ref{sec:case-study-replanning} showcases how the combined replanning and tracking procedure in Alg.~\ref{alg:finite-interval-method} can successfully be used to guide the AUV through an \textit{a priori} unknown environment, and how flexible replanning leads to performance improvements. 
Since Alg.~\ref{alg:finite-interval-method} requires $T_{\text{run}} \leq T_{\text{off}}$, in Sec.~\ref{sec:case-study-replanning} we show that Alg.~\ref{alg:periodic-method} can be applied to the AUV when $T_{\text{run}} \gg T_{\text{off}}$ by exploiting periodicity.}

\subsection{Modeling the AUV Subject to Wave Disturbances} \label{sec:auv-model}
In this section, we first provide a time-varying, state-dependent system model of an AUV in plane progressive waves derived based on \cite{Battista2015UnderwaterWaves}. This model is lifted from \cite[eq. 9]{Siriya2020Safety-GuaranteedWaves}, which we refer to as the tracking system model for Case 1. Following this, we derive two approximations of this tracking system model which differ based on the fidelity of the wave disturbance model. In particular, Case 2 corresponds to a time-varying, state-independent wave approximation, and Case 3 is a time- and state-independent approximation. These cases will be used to 
 demonstrate the benefits and limitations of including or discarding model information in our methods.

\subsubsection{Case 1: Tracking System for Time-Varying, State-Dependent Wave Model} \label{sec:auv-model-case-1}
Consider the following model $\dot{s} = \underbar{f}(s,u_s,d^{\text{nom}},\ssreview{d^{\text{wave}}})$ of an AUV subject to \ssreview{both} a nominal \ssreview{and wave} disturbance:
\begin{align}
    \dot{x} &= u_r + \ssreview{d^{\text{wave}}_x} + d_x \nonumber\\
    \dot{z} &= w_r + \ssreview{d^{\text{wave}}_z} + d_z \nonumber\\
    \dot{u}_r &= (m-X_{\dot{u}})^{-1}((\bar{m} - m)\ssreview{d^{\text{wave}}_u} \nonumber\\
    &\quad - (X_u + X_{\abs{u} u}\abs{u_r})u_r +  T_A) + d_u \nonumber\\
    \dot{w}_r &= (m-Z_{\dot{w}})^{-1} ((\bar{m} - m)\ssreview{d^{\text{wave}}_z}-(-g(m-\bar{m}))\\
    &\quad - (Z_w + Z_{\abs{w}w}\abs{w_r})w_r + T_B) + d_w. \label{eqn:auv-dynamics}
\end{align}
Here, $s = (x,z,u_r,w_r)$ is the state, where $(x,z)$ is the AUV position ($z$ pointing downwards) and $(u_r,w_r)$ is the relative velocity between the AUV and water flow. The control input is $u_s = (T_A,T_B)$, \ssreview{$d^{\text{wave}}=(d^{\text{wave}}_x,d^{\text{wave}}_z,d^{\text{wave}}_u,d^{\text{wave}}_z)$ is the wave disturbances,} and $d^{\text{nom}} = (d_x,d_z,d_u,d_w)$ \ssreview{is} other unmodeled, bounded disturbances. Model parameters include the AUV mass $m$, the displaced fluid \ssreview{mass} $\bar{m}$, the additive masses $X_{\dot{u}},Z_{\dot{w}}$, the linear damping factors $X_{u},Z_{u}$, the quadratic damping factors $X_{\abs{u}u},Z_{\abs{w}w}$, and gravity $g$. 

Let $\underbar{d}^{\text{tvsd}}(\ssreview{t,s}) = (W_x(x,z,t), W_z(x,z,t), A_x(x,z,t), \allowbreak A_z(x,z,t))$ represent the plane-progressive wave model, where \ssreview{$(W_x,W_z)$ and $(A_x,A_z)$} are the velocity and acceleration models, given by
\begin{align}
    \begin{aligned}
        \begin{bmatrix} 
            W_x(x,z,t) \\
            W_z(x,z,t)
        \end{bmatrix} = \tilde{A} \omega e^{-kz} \begin{bmatrix}
            \cos(kx-\omega t) \\
            -\sin(kx - \omega t)
        \end{bmatrix}
    \end{aligned},
    \label{eqn:wave-flow-velocity} \\
    \begin{bmatrix}
    A_x(x,z,t) \\
    A_z(x,z,t)
\end{bmatrix} = \tilde{A} \omega^2 e^{-kz} \begin{bmatrix}
        \sin(kx-\omega t) \\
        \cos(kx - \omega t)
    \end{bmatrix}. \label{eqn:wave-flow-acceleration}
\end{align}
Here, $\tilde{A}$ is the wave amplitude, $\omega$ is the wave frequency, and $k$ is the wavenumber. 
\ssreview{Note that $\underbar{d}^{\text{tvsd}}$ is both time-varying and state-dependent. The Case 1 tracking system model is then obtained by choosing the wave disturbance $d^{\textnormal{wave}}$ in the AUV model \eqref{eqn:auv-dynamics} as $\underbar{d}^{\textnormal{tvsd}}(s,t)$, producing a time-varying system model $f(t,s,u_s,d) = \underbar{f}(s,u_s,d^{\text{nom}}, \underbar{d}^{\text{tvsd}}(t,s))$ for $t \in [0,T]$, with $d = d^{\text{nom}}$.} 

\subsubsection{Case 2: Tracking System for Time-Varying, State-Independent Wave Model}  \label{eqn:case-2-model-procedure}
In this case, we consider a time-varying, state-independent approximation of the plane-progressive wave model $\underbar{d}^{\text{tvsd}}(t,s)$ with the following structure:
\begin{align}
    \underbar{d}^{\text{tvsi}}(t,d^{\text{tvsi}}) &= \begin{bmatrix}
        \tilde{A}^{\text{tvsi}}_W \cos(\phi_W-\omega t) + d_{W,x}^{\text{tvsi}}\\
        -\tilde{A}^{\text{tvsi}}_W\sin(\phi_W - \omega t) + d_{W,z}^{\text{tvsi}}\\
        \tilde{A}^{\text{tvsi}}_A\sin(\phi_A-\omega t) + d_{A,x}^{\text{tvsi}}\\
        \tilde{A}^{\text{tvsi}}_A\cos(\phi_A - \omega t) + d_{A,z}^{\text{tvsi}}
        \end{bmatrix}, \label{eqn:dtvsi}
\end{align}
\noeqref{eqn:dtvsi}
where $d^{\text{tvsi}}=(d_{W,x}^{\text{tvsi}}$,$d_{W,z}^{\text{tvsi}}$,$d_{A,x}^{\text{tvsi}}$,$d_{A,z}^{\text{tvsi}})$ are bounded uncertainties satisfying $\abs{d_{W,x}^{\text{tvsi}}},\abs{d_{W,z}^{\text{tvsi}}} \leq D_W^{\text{tvsi}}$ and $\abs{d_{A,x}^{\text{tvsi}}},\abs{d_{A,z}^{\text{tvsi}}} \leq D_A^{\text{tvsi}}$.
Here, $\tilde{A}^{\text{tvsi}}_W,\tilde{A}^{\text{tvsi}}_A,\phi_W,\phi_A,D_W^{\text{tvsi}},D_A^{\text{tvsi}}$ are parameters that need to be appropriately selected so that $\underbar{d}^{\text{tvsd}}(\ssreview{t,s})$ is conservatively approximated by $\underbar{d}^{\text{tvsi}}(t,d^{\text{tvsi}})$ over some bounded region of interest $\Xi \subseteq \mathbb{R}^2$ in the $x$-$z$ space, and time interval $[0,T]$. They are chosen so that for all $(x,z) \in \Xi$ and $t \in [0,T]$,
\begin{align}
    &\begin{bmatrix}
        W_x(x,z,t) \\
        W_z(x,z,t) \\
        A_x(x,z,t) \\
        A_z(x,z,t)
    \end{bmatrix}
    \in \left\{\begin{bmatrix}
        \tilde{A}^{\text{tvsi}}_W\cos(\phi_W-\omega t) \\
        \tilde{A}^{\text{tvsi}}_W(-\sin(\phi_W-\omega t)) \\
        \tilde{A}^{\text{tvsi}}_A\sin(\phi_A-\omega t) \\
        \tilde{A}^{\text{tvsi}}_A\cos(\phi_A-\omega t) 
    \end{bmatrix}\right\} \\
    & \quad \quad \oplus [-D_W^{\text{tvsi}},D_W^{\text{tvsi}}]^2 \times [-D_A^{\text{tvsi}},D_A^{\text{tvsi}}]^2 \label{eqn:case-study-offline-computation-tv-bound}
\end{align}
holds with the smallest $\tilde{A}^{\text{tvsi}}_{W},\tilde{A}^{\text{tvsi}}_{A},D_W^{\text{tvsi}},D_A^{\text{tvsi}} > 0$.

The \ssreview{Case 2} tracking system model is \ssreview{then formulated as the time-varying system} $f(t,s,u_s,d)=\underbar{f}(s,u_s,d^{\text{nom}},\underbar{d}^{\text{tvsi}}(t,d^{\text{tvsi}}))$ for $t \in [0,T]$, with $d = (d^{\text{nom}},d^{\text{tvsi}})$.

\subsubsection{Case 3: Tracking System for Time- and State-Independent Wave Model} \label{eqn:case-3-model-procedure}
For Case 3, the wave model is simply an additive, bounded disturbance $d^{\text{tisi}}=(d_{W,x}^{\text{tisi}},d_{W,z}^{\text{tisi}},d_{A,x}^{\text{tisi}},d_{A,z}^{\text{tisi}})$, satisfying $\abs{d_{W,x}^{\text{tisi}}},\abs{d_{W,z}^{\text{tisi}}} \leq D_W^{\text{tisi}}$ and $\abs{d_{A,x}^{\text{tisi}}},\abs{d_{A,z}^{\text{tisi}}} \leq D_A^{\text{tisi}}$. The parameters $ D_W^{\text{tisi}}, D_A^{\text{tisi}}$ are \ssreview{chosen as the smallest values capturing} uncertainty associated with moving from \ssreview{$\underline{d}^{\textnormal{tvsd}}$} to a time-invariant approximation, i.e. \ssreview{$\underline{d}^{\textnormal{tvsd}}(t,(x,z,u_r,w_r)) \in [-D_W^{\text{tisi}},D_W^{\text{tisi}}]^2 \times [-D_A^{\text{tisi}},D_A^{\text{tisi}}]^2$ must be satisfied for all $(x,z,u_r,w_r) \in \Xi \times \mathbb{R}^2$} and $t \in [0,T]$. The \ssreview{Case~3} tracking system model is then \ssreview{formulated as a time-invariant system $f(s,u_s,d)=\underbar{f}(s,u_s,d^{\text{nom}},d^{\text{tisi}})$ over} $t \in [0,T]$, with $d = (d^{\text{nom}},d^{\text{tisi}})$.

\subsection{Interfacing OfflineStage with the AUV Model} \label{sec:case-study-offline-computation}
In this section, we describe three different ways of selecting \ssreview{the} inputs $L, M, C, f, h$ for OfflineStage based on \ssreview{Cases~1-3 in \ref{sec:auv-model}}. In all cases, we choose the planning system $h$ \ssreview{as} the single integrator $\dot{p} = h(p,u_p)=(\dot{x}_p,\dot{z}_p)=(u_{p,x},u_{p,z})$,  where $p=(x_p,z_p)$ is the planning system state and $u_p=(u_{p,x},u_{p,z})$ is the planning system control.

\subsubsection{Case 1: Full Wave}  \label{sec:case-study-offline-computation-full-wave}
Recall the Case 1 tracking system model $f(\ssreview{t,}s,u_s,d) = \underbar{f}(s,u_s,d^{\text{nom}}, \underbar{d}^{\text{tvsd}}(t,s))$, with $d = d^{\text{nom}}$. Here, we choose $L = \begin{bmatrix} I_2 & 0_{2 \times 2} & I_2 \\ 0_{2 \times 2} & I_2 & \ssupdated{0_{2 \times 2}} \end{bmatrix}^{\top}$ and $M = \begin{bmatrix} I_2 & 0_{2 \times 4} \end{bmatrix}^{\top}$, which following \eqref{eqn:time-varying-relative-model} formulates a 6D relative system with states $r=(x_{\alpha},z_{\alpha},u_r,w_r,x,z)$, where $(x_{\alpha},z_{\alpha})=(x-x_p,z-z_p)$ is the tracking error between the tracking and planning positions. The position $(x,z)$ is needed in $r$ since $\underbar{d}^{\text{tvsd}}(\ssreview{t,s})$ depends on it, in contrast to \ssupdatedd{Cases} 2 and 3. This produces different value functions, requiring different online planning and tracking procedures to handle this change. The error matrix C is then chosen as $C = \begin{bmatrix} I_2 & 0_{2 \times 4} \end{bmatrix}$.
\ssreview{Moreover, line 3 in Alg.~\ref{alg:offline-stage} was carried out using OptimizedDP \cite{bui2022optimizeddp} due to its 6D support.}

\subsubsection{Case 2 and 3: Time-Varying/Time-Invariant Wave Approximation} \label{sec:case-study-offline-computation-tv-wave}
Recall the Case 2 tracking system model $f(t,s,u_s,d)=\underbar{f}(s,u_s,d^{\text{nom}},\underbar{d}^{\text{tvsi}}(t,d^{\text{tvsi}}))$, with $d = (d^{\text{nom}},d^{\text{tvsi}})$, and the Case 3 model \ssreview{$f(s,u_s,d)=\underbar{f}(s,u_s,d^{\text{nom}},d^{\text{tisi}})$}, with $d = (d^{\text{nom}},d^{\text{tisi}})$. For both cases, we choose $L = I_4$ and $M = \begin{bmatrix} I_2 & 0_{2 \times 2} \end{bmatrix}^{\top}$, which following \eqref{eqn:time-varying-relative-model} formulates a 4D relative system with states $r=(x_{\alpha},z_{\alpha},u_r,w_r)$. Unlike Case 1, $(x,z)$ is not needed in $r$ since both $\underbar{d}^{\text{tvsi}}(t,d^{\text{tvsi}})$ and $d^{\text{tisi}}$ are state-independent. $C$ is then chosen as $C = \begin{bmatrix} I_2 & 0_{2 \times 2} \end{bmatrix}$. Note that the offline computation of Case 3 can be viewed as a \ssupdated{naive} extension of the offline procedure in FaSTrack to time-varying systems by uniformly bounding the waves.

\begin{remark} \label{remark:waves-adversarial}
    \ssreview{Although we pointed out in Sec.~\ref{sec:hj-game} that the consideration of worst-case planning system \ssupdatedd{behavior} is usually conservative for most problems, we find it is reasonable in this example. Here, it corresponds to the tracking system $f$ trying to catch up to a planning system $h$ that opposes the direction of the waves in \eqref{eqn:wave-flow-velocity}, which can occur frequently due to the periodic nature of these waves.}
\end{remark}

\subsection{Trajectory Planning for the AUV} \label{sec:case-study-planning}

Although Alg. \ref{alg:planning-bounded-interval} abstractly describes the planning procedure, how to implement the procedure remains unclear since we need to compute objects that are defined in continuous time and space. In this section, we describe example implementations of Alg. \ref{alg:planning-bounded-interval} for the problem of navigating an AUV subject to plane-progressive wave disturbances from an initial state $s_0 \in \mathcal{S}$ to a goal region $\mathcal{G}$ while satisfying constraints $\mathcal{C}$. These implementations are based on the Case 1-3 AUV models from \ref{sec:auv-model} and were originally contained in our conference paper \cite{Siriya2020Safety-GuaranteedWaves}. 
\ssreview{Note that the objects computed in these implementations are discrete time and space approximations of the continuous objects in Alg.~\ref{alg:planning-bounded-interval}, and thus not directly compatible with the safety guarantees of this work. Nevertheless, they are practically useful, and we leave considering the effect of approximation errors on guarantees \ssupdatedd{for} future work.}
We assume \ssreview{the AUV is a point mass for simplicity}, and that \ssreview{there are only static} obstacles $\mathbb{O}\subseteq \mathbb{R}^2$ in the $x \text{-}z$ plane, such \ssreview{that we have a static constraint} $\mathcal{C}=\mathbb{O}^{\mathsf{c}} \times \mathbb{R}^2$. \ssreview{However, considering the AUV volume is important for collision avoidance in reality, and it can easily be accounted for using the constraint $\mathcal{C}=(\mathbb{O} \oplus (-\mathbb{V}))^{\mathsf{c}} \times \mathbb{R}^2$, where $\mathbb{V}$ is the AUV volume centered at the tracking system position.} We also assume $\Xi^{\mathsf{c}} \subseteq \mathbb{O}$ \ssminor{since} the wave models are only accurate for $(x,z) \in \Xi$\ssminor{, and} that the goal is determined by position, i.e. $\mathcal{G}= \mathbb{G}\times \mathbb{R}^2 \subseteq \mathcal{S}$, where $\mathbb{G}\subseteq \mathbb{R}^2$ is the goal region in the $x\text{-}z$ plane.
In 
\ref{sec:case-study-planning-implementation}, we describe how to approximately implement the algorithm for the models based on Cases 1-3. In \ref{sec:case-study-planning-simulation}, we briefly comment on the differences.

\subsubsection{Implementing Planning Algorithms} \label{sec:case-study-planning-implementation}

We first describe the implementation of Alg. \ref{alg:planning-bounded-interval} for Case 1, followed by 2 and 3 simultaneously. Throughout all examples, suppose $T_{\text{off}} \in [0,T]$, $t_i \in [0,T_{\text{off}}]$, and $t_f \in [t_i,T_{\text{off}}]$. Moreover, suppose $\mathcal{C}$ and $\mathcal{G}$ have been provided, $h$ is the single integrator system from \ref{sec:case-study-offline-computation}, and both $c \geq \underbar{V}(s,t)$ and $p \in \mathcal{T}_p(s,t_{\text{i}};c)$ are given.

\paragraph{Case 1: Full Wave Model}
Suppose Alg. \ref{alg:offline-stage} has been run for Case 1 based on \ref{sec:case-study-offline-computation-full-wave}. Let $\oplus$ denote the Minkowski sum. To implement line 2 of Alg. \ref{alg:planning-bounded-interval}, the planning system constraints can be expressed as $\mathcal{C}_p(t) = \mathcal{F}_{\mathcal{C}}(t;c) = \mathcal{O}_p^{\mathsf{c}}(t)$, where $\mathcal{O}_p$ denotes the \textit{planning system obstacles} and is defined using the set avoidance map from \ssminor{Sec.} \ref{sec:sets} as
\begin{align}
    \mathcal{O}_p(t) = \mathcal{A}_{\mathcal{C}^{\mathsf{c}}}(t;c) &\ssminor{=\bigcup_{(x,y) \in \mathbb{O}} \bigcup_{(u_r,w_r) \in \mathbb{R}^2} \mathcal{T}_p(s,t;c)} \\
    &=  \bigcup_{(x,z) \in \mathbb{O} } \{(x,z)\} \oplus (- \mathcal{B}_e^{\text{full}}(x,z,t;c) ) 
\end{align}
for $t \in [t_{\text{i}},t_{\text{f}}]$, where $\mathcal{B}_e^{\text{full}}(x,z,t;c) := \{ (x_{\alpha},z_{\alpha}) :  V((x_{\alpha},z_{\alpha},u_r,w_r,x,z),t) \leq c \text{ for some } u_r,w_r \}$ denotes the tracking error bound (TEB) for the full wave model. We can approximately compute $\mathcal{O}_p(t)$ by setting up a grid of discrete cells representing $x,z$ points to iterate over, and following Alg. \ref{alg:computing-planner-obstacles}. In Alg. \ref{alg:computing-planner-obstacles}, $k$ refers to the time step, $N$ is the horizon length, $t_k=t_{\text{i}}+kT_s$ is the discretized time, and $T_s = (t_{\text{f}}-t_{\text{i}})/N$ is the time discretization period. Note that $\mathcal{C}_p(t_k)$ could be obtained by taking the complement of the output $\mathcal{O}_p(t_k)$. We will keep $\mathcal{O}_p(t_k)$ since it is a more convenient representation in this example.

\begin{algorithm}
\caption{Computing Planning System Obstacles for Case 1 AUV Example} \label{alg:computing-planner-obstacles}
\begin{algorithmic}[1]
    \STATE \textbf{Input}: $\mathbb{O}$, $N$, $\mathcal{B}_{e}^{\text{full}}(x,z,t_k;c) \text{ for } k \in \{0,\hdots,N\}$, $t_{\text{i}}$, $T_s$
    \FOR{$k \in \{0,\hdots,N\}$}
        \STATE $t_k = t_{\text{i}} + kT_s$
        \STATE $\mathcal{O}_p(t_k) \leftarrow \emptyset$
        \FOR{$(x,z) \in \mathbb{O}$}
             \FOR{$(x_{\alpha},z_{\alpha})\in \mathcal{B}_{e}^{\text{full}}(x,z,t_k;c)$}
                \STATE $\mathcal{O}_p(t_k) \leftarrow \mathcal{O}_p(t_k) \cup \{(x,z)-(x_{\alpha},z_{\alpha}) \}$
            \ENDFOR
        \ENDFOR
    \ENDFOR
    \STATE \textbf{Output}: $\mathcal{O}_p(t_k) \text{ for } k \in \{0,\hdots,N\}$
\end{algorithmic}
\end{algorithm}

Similarly, $\mathcal{G}_p(t), \ t \in [t_{\text{i}},t_{\text{f}}]$ can be formulated as $\mathcal{G}_p(t) = (\bigcup_{(x,z)\in\mathbb{G}^{\mathsf{C}}} \{(x,z)\} \oplus (-\mathcal{B}_{e}^{\text{full}}(x,z,t;c)))^{\mathsf{C}}$, so line 3 of Alg. \ref{alg:planning-bounded-interval} is approximately implemented by inputting $\mathcal{G}$ instead of $\mathcal{O}$ into Alg. \ref{alg:computing-planner-obstacles}, and taking the complement of the output.

Finally, we implement line 4 of Alg. \ref{alg:planning-bounded-interval} by solving a discrete-time trajectory optimization problem in \eqref{eqn:mpc-problem}, where $h_d$ denotes the Euler discretization of $h$ with period $T_s$. Conditions C1-C4 are represented in the constraints of this problem.
\begin{align}
    &\min_{\{p_k\}_{k=0}^{N}, \{ u_{p,k} \}_{k=0}^{N-1}} \sum_{k=0}^{N-1}\big( (p_k-p_{\text{ref}})^{\top} Q (p_k-p_{\text{ref}}) \nonumber \\
    & \quad \quad \quad + u_{p,k}^{\top}R u_{p,k} \big) + (p_N-p_{\text{ref}})^{\top}Q(p_N-p_{\text{ref}})   \\
    &\text{s.t.} \quad p_{k+1} = h_d(p_k,u_{p,k}), \ u_{p,k} \in \mathcal{U}_p, \ p_k \not \in \mathcal{O}_p(t_k), \\
    & \quad \quad p_0 = p, \ p_N \in \mathcal{G}_p(t_N). \label{eqn:mpc-problem}
\end{align}
Here, $\{p_k\}_{k=0}^N$ and $\{u_{p,k}\}_{k=0}^{N-1}$ are planning system states and controls over the horizon respectively, and $p_{\text{ref}}\in \mathbb{R}^2$ is a reference planning system state chosen to satisfy $p_{\text{ref}} \in \mathcal{G}_p(t_N)$. The weight matrices $Q$ and $R$ penalize the tracking error and control inputs. The obstacle avoidance constraint $p_k \not \in \mathcal{O}_p(t_k)$ can be practically implemented using optimization-based collision avoidance constraints \cite{Zhang2017Optimization}.

\ssreview{Note} that \ssreview{here}, the computation of the planning system obstacles and \ssupdatedd{the} goal \ssreview{is} the same as \cite[IV.B]{Siriya2020Safety-GuaranteedWaves}, \ssreview{differing} from FaSTrack \cite{Chen2021FaSTrack:Tracking}. \ssreview{This enables the capturing of} spatial information in the wave model.

\paragraph{Case 2 and 3: Time-Varying/Invariant Wave Approximation} \label{sec:planning-example-case-2}
Suppose Alg. \ref{alg:offline-stage} has been run for Case 2 or 3 based on \ref{sec:case-study-offline-computation-tv-wave}. Unlike in Case 1, $\mathcal{O}_p(t_k)$ can be simply formulated as $\mathcal{O}_p(t_k) = \mathbb{O} \oplus (-\mathcal{B}_e^{\text{approx}}(t_k;c))$ for $k \in \{0,\hdots,N\}$, where $\mathcal{B}_e^{\text{approx}}(t_k;c) := \{(x_{\alpha},z_{\alpha}) : V((x_{\alpha},z_{\alpha},u_r,w_r),t) \leq c \text{ for some } u_r,w_r \in \mathbb{R} \}$ is the TEB for the approximate wave model. $\mathcal{B}_e^{\text{approx}}(t;c)$ can be computed using ToolboxLS \cite{Mitchell2005ASystems}, and further approximated with the minimum bounding ball for convenience. Its implementation is simpler than Case 1 since $\mathcal{B}_e^{\text{approx}}(t;c)$ is independent of $x,z$. The planning system goal is also easily obtained via $\mathcal{G}_p(t_k)= \mathbb{G} \ominus \mathcal{B}_e^{\text{approx}}(t_k;c)$, where $\ominus$ denotes the Minkowski difference. In these examples, the computations of the planning system obstacles and goals are the same as \cite[IV.A]{Siriya2020Safety-GuaranteedWaves}, and coincide with FaSTrack \cite{Chen2021FaSTrack:Tracking}. A trajectory that avoids planning system obstacles and enters the goal region is then computed in line 4 by solving the MPC problem in \eqref{eqn:mpc-problem} using the new $\mathcal{O}_p(t_k)$ and $\mathcal{G}_p(t_k)$.

\subsubsection{Simulating Planning Algorithms} \label{sec:case-study-planning-simulation}
This section aims to briefly explain the difference in performance between the planned trajectories across Cases 1, 2, and 3. We qualitatively describe these differences based on the numerical example in our conference paper \cite[\ssminor{Sec.} V]{Siriya2020Safety-GuaranteedWaves}, which we refer the reader to for full details. Note that Cases 1, 2, and 3 in this paper correspond to Cases 3, 2, and 1 (reversed order) in \cite{Siriya2020Safety-GuaranteedWaves}.
\ssreview{We also briefly discuss the importance of choosing an appropriate case for real-time implementation.}

Our findings were that the planned trajectory for Case 1 reaches the planning system goal region in the shortest time, followed by Case 2, and lastly Case 3. This order is explained by the fact that as we move from Case 1 to Case 2, spatial information for $\underbar{d}^{\text{tvsd}}(\ssreview{t,s})$ is moved to the unmodelled uncertainty in $d^{\text{tvsi}}$, and both time \textit{and} spatial information is captured in $d^{\text{tisi}}$ for Case 3. The greater the magnitude of the unmodelled uncertainty, the larger the planning system obstacles and the smaller the goal, producing conservative trajectories and slower goal-reaching. These results highlight the advantages of our time-varying method, compared to \ssupdated{naive}ly extending FaSTrack by uniformly bounding the disturbances.

Despite the benefit of using \ssreview{Case 1, there are computational issues to consider. Firstly, OfflineStage in Alg.~\ref{alg:offline-stage} is more computationally expensive} due to the curse of dimensionality in dynamic programming\ssreview{, since it involves a 6D model, versus Cases 2 and 3 which are 4D. \ssupdatedd{For} OfflineStage to be tractable in Case 1, a coarser discretization is required, decreasing the accuracy of the approximated value function. Moreover, the generation of planning system obstacles also becomes more complex and expensive (see Alg.~\ref{alg:computing-planner-obstacles} versus Sec.~\ref{sec:planning-example-case-2}), which is important because this computation is performed online. Thus, it is critical to carefully consider the choice of model with the view to attain real-time performance.} \ssreview{This is further discussed in \cite{Siriya2020Safety-GuaranteedWaves} which focuses on the planning problem.}

\subsection{Replanning and Tracking for the AUV} \label{sec:case-study-replanning}

In this section, we \ssminor{showcase} an implementation of Alg. \ref{alg:finite-interval-method} on a simulation of the \ssminor{plane-progressive wave} AUV \ssminor{model obtained by combining} 
\eqref{eqn:auv-dynamics} \ssminor{with \eqref{eqn:wave-flow-velocity}}.
\ssreview{The AUV model parameters we consider \ssupdated{are obtained from \cite{Proctor2014Semi-AutonomousROV}, and are $m=116$ kg, $Z_{\dot{w}}=-383$ kg, $X_{{\lvert u \rvert}u}=241.3$ kg/m, $\bar{m}=116.2 kg$, $X_{u}=26.9$ kg/s, $Z_{{\lvert w \rvert} w}=265.6 kg/m$, $X_{\dot{u}}=-167.7$ kg, $Z_{w}=0$ kg/s, $m=116$ kg, $Z_{\dot{w}}=-383$ kg, $X_{{\lvert u \rvert}u}=241.3$ kg/m, $\bar{m}=116.2$ kg, $X_{u}=26.9$ kg/s, $Z_{{\lvert w \rvert} w}=265.6$ kg/m, $X_{\dot{u}}=-167.7$ kg, $Z_{w}=0$ kg/s. Moreover, the disturbances } $d^{\textnormal{nom}}=(d_x,d_z,d_u,d_w)$ are each obtained by applying zero-order hold with time interval $0.2 \ \textnormal{s}$ to discrete-time i.i.d. random sequences sampled from a \ssupdatedd{uniform} distribution over $[-0.001,0.001]$.}
The plane-progressive wave parameters for \eqref{eqn:wave-flow-velocity} we consider are \ssupdated{$\tilde{A}=0.4$ m, $\omega=2 \pi (0.1)$ rad/s, and $k=0.0402$ rad/m.}

\ssupdatedd{We aim} to control the robot so that starting from $s_0=(-1.4,2.74,0.0,0.0)$, it reaches the goal region $\mathcal{G}$ whilst avoiding all obstacles in an \textit{a priori} unknown static environment over the region $\Xi = \{(x,z) \mid x \in [-2,2], z \in [2,6] \}$. We consider $T_{\text{run}}=8$ seconds, and $T_{\text{off}} = 10$ seconds. In this task, \ssreview{we consider an online constraint $\mathcal{C}(t)$ (described in Sec.~\ref{sec:online-constraints}) representing} known information about the true obstacles at time $t$. We assume that the AUV is equipped with a sensor with a range of \ssupdated{$1.2 \ \textnormal{m}$}, and whenever an obstacle partially enters this sensor range at time $t$, the entirety of the obstacle has been sensed, so $\mathcal{C}(t)$ is updated to include this new obstacle information. Additionally, $\Xi^{\mathsf{c}}$ is considered as a known obstacle for all $t$. Moreover, for the HJ game, we assume the \ssupdated{tracking/planning system inputs and disturbances are constrained so $\vert T_A \vert, \vert T_B \vert \leq 1000$ N, $\vert u_{p,x} \vert, \vert u_{p,z} \vert \leq 0.3$ m/s, $\vert d_x \vert, \vert d_z \vert \leq 0.001$ m/s, and $\vert d_u \vert, \vert d_w \vert \leq 0.001$ m/$\textnormal{s}^2$ are all satisfied.} \ssreview{For this example, we consider the Case~2 tracking system model from \ref{eqn:case-2-model-procedure} for Alg.~\ref{alg:finite-interval-method}, whose \ssupdated{wave model} parameters were chosen \ssupdated{as $\tilde{A}^{\textnormal{tvsi}}_W=0.2185$ m/s, $\tilde{A}^{\textnormal{tvsi}}_A=0.1373$ m/$\textnormal{s}^2$, $\phi_W=\phi_A=0$ rad, $D_W=0.03$ m/s, and $D_A=0.025$ m/$\textnormal{s}^2$,} since they produce the smallest $\tilde{A}^{\text{tvsi}}_{W},\tilde{A}^{\text{tvsi}}_{A},D_W^{\text{tvsi}},D_A^{\text{tvsi}}>0$ satisfying \eqref{eqn:case-study-offline-computation-tv-bound}.} Hence, the inputs for OfflineStage in line 2, as well as the implementation of PlanTrajectoryOverInterval in line 15, both correspond to Case 2, as described in \ref{sec:case-study-offline-computation-tv-wave} in \ref{sec:planning-example-case-2}. \ssreview{Note that although Case~1 is also compatible, we stick to Case~2 in this section to focus our discussion on the ability of Alg.~\ref{alg:finite-interval-method} to handle time-varying systems, rather than the benefit of including spatial information as previously discussed in Sec.~\ref{sec:case-study-planning-simulation}.}

In Alg. \ref{alg:finite-interval-method}, we have the flexibility of choosing when to replan in line 11, and also the value level $c_k$ and planning system state $p_k$ according to lines 14 and 15. We repeat the simulation twice under two different implementations of these lines. In Simulation 1, we replan if an unknown obstacle enters the sensor range of the AUV, and fix $c_k = \underbar{V}(s_0,0)=0.61$, and $p_k=\hat{\xi}_{h,k-1}(t_{k})$. In Simulation 2, in addition to replanning when an unknown obstacle enters the sensor range, we replan every $2$ seconds \ssreview{so that the teleportation phenomenon described in Remark~\ref{remark:replan-level} occurs more frequently. Additionally, we replan} when the AUV enters the region defined by \ssupdated{$z > 3.6 \ \textnormal{m}$}. Moreover, we initialize $p_k$ at the closest point in $\mathcal{T}_p(s_k,t_k;c_k)$ to $\mathcal{G}_p(t_k)$. Initially, we set $c_k = \underbar{V}(s_0,0) = 0.61$, but change it to $c_k=\ssminor{0.7}$ when the AUV enters this region. The reason for this is to increase the size of the region $\mathcal{T}_p(s_k,t_k;c_k)$ over which $p_k$ is reinitialized \ssreview{so that the effects of teleportation are further accentuated.}

\begin{figure}[h]
    \centering
    \includegraphics[width=0.4\textwidth]{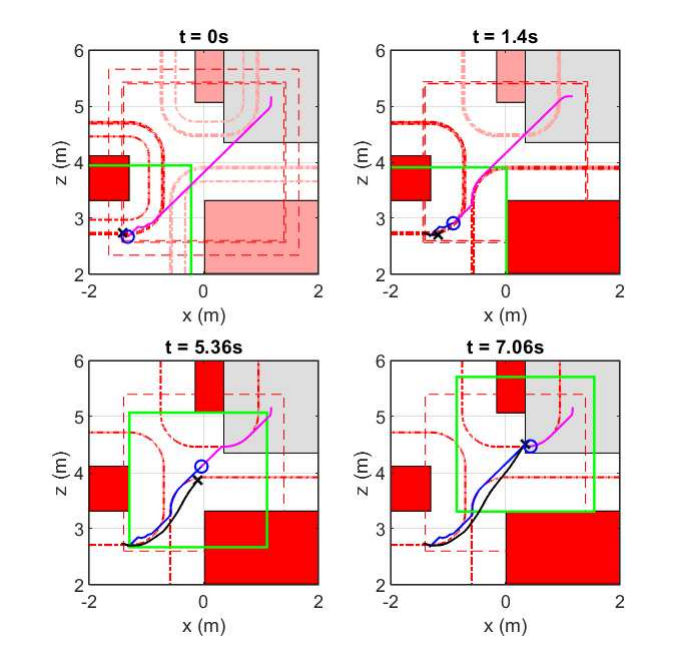}
    \caption{Snapshots of the tracking and planning systems at time instances where replanning occurs up until the goal is reached, for Simulation 1. Dark red blocks correspond to known obstacles, and the corresponding dark red dotted lines correspond to the planning system obstacles. The dark red dashed squares are due to $\Xi^{\mathsf{c}}$, which is included in the known obstacles. Light red blocks correspond to unknown obstacles, and the corresponding light red dotted lines are the unknown future planning system obstacles. The grey box is the goal region. The black `x' is the tracking system position, and the black line is the tracking system trajectory up to the current time. The green box is the sensor range for the tracking system. The blue circle is the planning system, the blue line is the planning system trajectory up to the current time, and the purple line is the current planned trajectory.}
    \label{fig:case_study_tv_noswitch}
\end{figure}

Snapshots of the trajectory for Simulations 1 and 2 and at various time instances where replanning occurs are shown in Figures \ref{fig:case_study_tv_noswitch} and \ref{fig:case_study_tv_switch} respectively. In Simulation 1, replanning due to obstacle detection occurs at $t=0,1.4$ and $5.36$ seconds, and the goal is reached at $7.06$ seconds. In Simulation 2, replanning due to obstacle detection occurs at $t=0,1.1$ and $3.8$ seconds, and the goal is reached at $4.2$ seconds. Replanning due to $z > 3.6$ occurs at $t=3.8$ seconds.

\begin{figure}[h]
    \centering
    \includegraphics[width=0.4\textwidth]{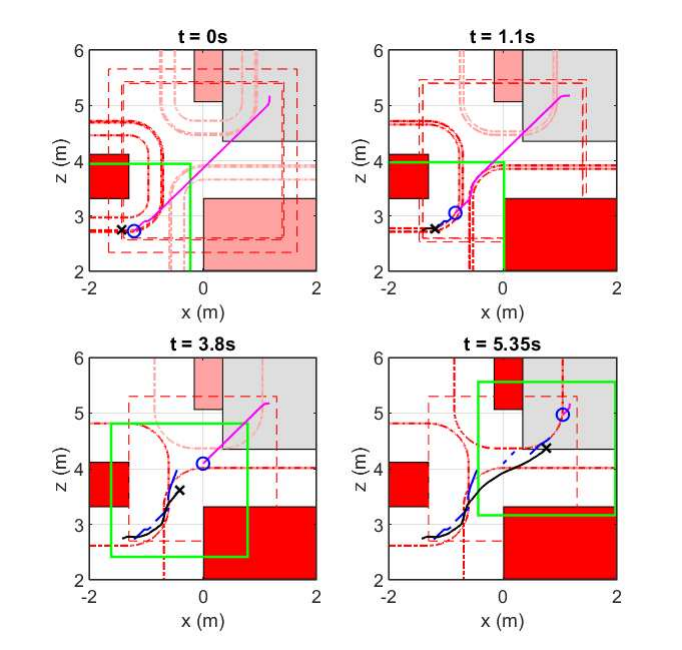}
    \caption{Snapshots of the tracking and planning system at times where replanning occurs up until the goal is reached for Simulation 2. The meaning of symbols and lines are consistent with Fig. \ref{fig:case_study_tv_noswitch}.}
    \label{fig:case_study_tv_switch}
\end{figure}
In both cases, we observe that whenever an obstacle comes into the sensor range for the tracking system such that $\mathcal{C}(t)$ changes, a new trajectory is replanned which avoids the corresponding planning system obstacle. This enables the tracking system to successfully avoid all obstacles and reach the goal region, despite the environment being \textit{a priori} unknown. Moreover, the tracking system in Simulation 2 arrives at the goal region earlier than \ssupdatedd{in} Simulation 1. This can be attributed to the differences in the implementation of the replanning lines. In particular, due to \ssupdatedd{the} selection of a larger $c_k$ when the AUV satisfies \ssupdated{$z > 3.6 \ \textnormal{m}$} and the initialization of $p_k \in \mathcal{T}_p(s_k,t_k;c_k)$, the planning system \ssupdatedd{can} ``teleport'' in Simulation 2. This phenomenon can be seen via the discontinuities of the planning system trajectory in Fig. \ref{fig:case_study_tv_switch}. It effectively enables the planning system to traverse through the environment beyond the speed specified when precomputing the $V(r,t)$ and $u_s^*(r,t)$, encouraging the tracking system to catch up and therefore result in faster goal-reaching. This is useful, but note that $c_k$ can only be increased if there is enough space for the AUV to reach the goal. This is evident at $t=3.8$ seconds in Fig. \ref{fig:case_study_tv_switch}, where after increasing $c_k$, the planning system obstacles enlarge, leaving no gap for the planning system to pass around the bottom two obstacles if it had not already. 

These observations highlight a tradeoff between the degree of ``teleportation'' and the size of the planning system obstacles that can be exploited for desired performance. It allows our replanning mechanism to play a role \textit{similar} to ``meta-planning" for FaSTrack \cite{Fridovich-Keil2018PlanningPlanning}, which enables safe switching between fast planners for open environments, and slow planners for tight ones. Our method can obtain similar \ssupdatedd{behavior} to a fast planner via the selection of a higher $c_k$, and a slow planner via the selection of a lower $c_k$, with the benefit of only requiring \ssupdatedd{the} precomputation of a \textit{single} value function.

\ssreview{To obtain a better understanding of how conservative our method is, we also plotted the value level $c_k$ versus the evaluation of the value function on the relative state in Fig.~\ref{fig:safety-plot}. This plot illustrates how close the relative system trajectory is to breaching the sublevel set of $V(r,t)$ corresponding to $c_k$, which would cause the violation of safety guarantees. We see that the system comes close to violating the safety guarantees during the times when $c_k=0.61$, suggesting that our method is not too conservative for the considered problem. Although we pointed out in Sec.~\ref{sec:hj-game} that the consideration of worst-case planning system \ssupdatedd{behavior} is usually conservative for most problems, we find it is reasonable in this example. It can take place when the tracking system $f$ is catching up to a planning system $h$ fighting against the waves in \eqref{eqn:wave-flow-velocity}, which can occur frequently due to the periodic nature of the waves.}

\begin{figure}
    \centering
    \includegraphics[width=0.45\textwidth]{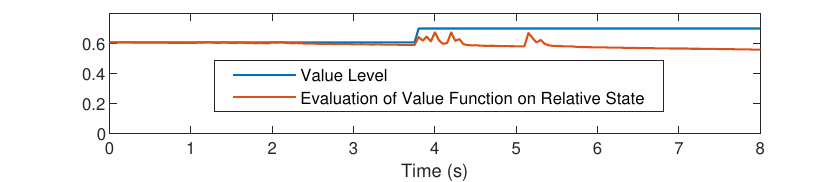}
    \caption{Comparison plot of value level $c_k$ versus the evaluation of the value function on the relative state for Simulation 2.}
    \label{fig:safety-plot}
\end{figure}

\begin{remark} \label{remark:stress-test}
    \ssreview{Fig.~\ref{fig:safety-plot} suggests our method in Simulation~2 is reasonable, but usually, considering worst-case disturbances/planning is conservative in applications. To get a feel for this, one could gradually increase the size of disturbances and planning system inputs in simulations while plotting the value level $c_k$ versus the evaluation of the value function (like in Fig.~\ref{fig:safety-plot}), until the latter breaches the former.}
\end{remark}

We also tried to simulate replanning and tracking when OfflineStage and PlanTrajectoryOverInterval are implemented according to Case 3, where the tracking system is time-invariant like FaSTrack. 
\ssreview{Following the modeling procedure in \ref{eqn:case-2-model-procedure}, the Case~3-specific parameters were selected as $D_W=0.2319$ and $D_A=0.1457$.}
However, the task was \ssreview{unsuccessful}. Even with the smallest value level $c_k = \underbar{V}(s_0,0)=0.67$, the planning system obstacles were too large, rendering planning infeasible. \ssupdated{This is consistent with the findings in Sec.~\ref{sec:case-study-planning-simulation}, further highlighting}  the benefits of incorporating time-varying information inside the tracking system dynamics and using Alg. \ref{alg:finite-interval-method}, as opposed to treating it as a bounded disturbance and modeling the tracking system as time-invariant.

\ssreview{The planning procedure in \eqref{eqn:mpc-problem} was solved using IPOPT \cite{wachter2006implementation} on a laptop with a Core I7-8565U CPU, with a discretization interval of $0.2$ s. In Simulation~2, the planning horizon varied from $N=104$ at $t=0$ to $60$ over the simulation based on the length of the planning interval $[t,T_{\textnormal{run}}]$ at the current time $t$. Moreover, the average solve time \ssupdated{for Simulation 1 was $1.0$ s (averaged over 4 solves, with a standard deviation of $0.4$ s), and for Simulation 2 was $1.4$ s (averaged over 9 solves, with a standard deviation of $0.9$ s).}
\ssupdated{These numbers cannot} yet be considered real-time but may be improved with faster hardware. \ssupdated{They can also be} improved by removing the goal constraint $p_N \in \mathcal{G}_p(t_N)$ in \eqref{eqn:mpc-problem} and shortening the planning interval in line 15 of Alg.~\ref{alg:finite-interval-method}, as discussed in Remark~\ref{remark:conditions}, and additionally applying different tricks for more efficiently solving optimization problems such as warm-starting.}

\begin{remark}
    \ssreview{One interesting question is how Alg.~\ref{alg:finite-interval-method} compares with FaSTrack \cite{Chen2021FaSTrack:Tracking} on a time-invariant system example. Note that the primary difference between Simulations~1 and 2 is that the former can ``teleport'', which was pointed out in Remark~\ref{remark:choose-replan} as the significant difference between Alg.~\ref{alg:finite-interval-method} and FaSTrack (outside of the former supporting time-varying systems). Thus, we expect that implementing Alg.~\ref{alg:finite-interval-method} similarly to Simulation~2 for a time-invariant system should produce faster \ssupdatedd{goal-reaching} compared to FaSTrack.}
\end{remark}

\subsection{Periodic Replanning and Tracking for the AUV} \label{sec:case-study-replanning-periodic}

In this section, we demonstrate an implementation of Alg. \ref{alg:periodic-method} on a simulation of the same AUV subject to plane-progressive waves as in \ssminor{Sec.} \ref{sec:case-study-planning} and \ref{sec:case-study-replanning}. We aim to control the robot starting from $s_0=(-3.3,2.7,0.0,0.0)$ so that it avoids all obstacles in the environment, which are \textit{a priori} unknown, and reaches successive goal regions $\mathcal{G}_1$, $\mathcal{G}_2$, and $\mathcal{G}_3$, which are illustrated in Fig. \ref{fig:case_study_periodic}. In particular, as soon as the AUV hits $\mathcal{G}_1$, the active goal is changed to $\mathcal{G}_2$, and then when it hits $\mathcal{G}_2$, the active goal becomes $\mathcal{G}_3$. Unlike the previous example, we enlarge the environment to the region $\{(x,z) \in \mathbb{R}^2 \mid x \in [-4,4], z \in [2,6] \}$. To accommodate this large environment and varying \ssupdatedd{goals}, we allow $T_{\text{run}}=60$ seconds, set $T_{\text{off}} = 14$ seconds, and $T'=4$ seconds. Notice that $T_{\text{run}} \gg T_{\text{off}}$. The sensor range of the AUV is \ssupdated{$1.2 \ \textnormal{m}$}, and again, the online, updated obstacles $\mathcal{O}(t)$ represents known obstacle information at time $t$. The inputs for OfflineStage in line 2, as well as the implementation of PlanTrajectoryOverInterval in line 15, both correspond to Case 2, as described in \ref{sec:case-study-offline-computation-tv-wave} in \ref{sec:planning-example-case-2} respectively.

In addition to replanning when $t - t_k \geq T'$, we choose to replan whenever a new obstacle enters the sensor range of the AUV, as well as the moment it hits the current active goal region. We also fix $c_k = 0.71$ and set $p_k$ as the closest point in $\mathcal{T}_p(s_k,t_k;c_k)$ to $\mathcal{G}_p(t_k)$ which avoids $\mathcal{O}_p(t_k)$. Note that we relax the planning algorithm so that C4 is not enforced in Alg. \ref{alg:planning-bounded-interval} by removing the MPC goal constraint in \eqref{eqn:mpc-problem} since the planning horizon is too small to plan a trajectory that reaches the planning system goal region.

Snapshots of the simulation at time instances where the active goal region has been reached \ssupdatedd{are} shown in Fig. \ref{fig:case_study_periodic}. The AUV successfully avoids all obstacles and reaches all of the active goal regions in succession, despite the environment being \textit{a priori} unknown. Additionally, the algorithm is successfully applied up until the final active goal is reached at $t=52.25$ seconds, which is significantly longer than $T_{\text{off}}=14$ seconds. This is achieved since Alg. \ref{alg:periodic-method} exploits the periodic properties of the system.

\begin{figure}
    \centering
    \includegraphics[width=0.49\textwidth]{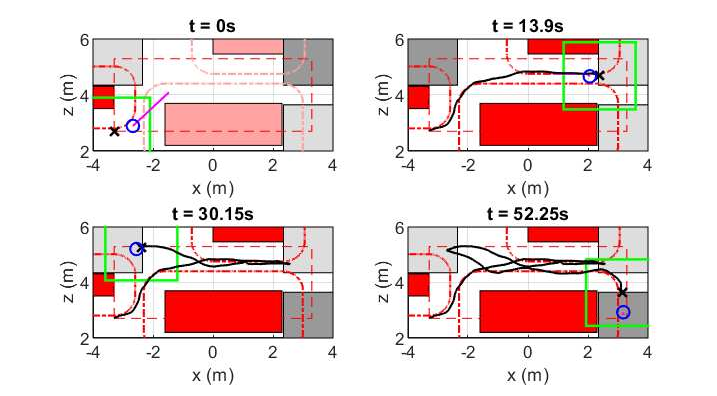}
    \caption{Snapshots of the tracking and planning systems at time instances where the replanning occurs due to the AUV reaching the active goal region. The meaning of symbols and lines are consistent with Fig. \ref{fig:case_study_tv_noswitch}, with the difference being that the dark grey box is the current active goal region, and the light grey boxes are the inactive goals.}
    \label{fig:case_study_periodic}
\end{figure}
\section{Conclusion} \label{sec:conclusion}

In this paper, we proposed a method for safety-guaranteed trajectory replanning and tracking for generally time-varying systems while satisfying constraints updated online over a finite horizon. Our procedure \ssupdated{is promising for real-time applications and} allows planned trajectories to `teleport', enabling \ssminor{flexible} performance across diverse environments. This was modified to obtain an algorithm \ssminor{applicable} over arbitrarily long time interval\ssminor{s} for periodic systems. We \ssupdatedd{analyze} the constraint satisfaction and \ssupdatedd{goal-reaching} guarantees for these methods and demonstrate \ssupdatedd{their} utility on \ssminor{an example} AUV subject to wave disturbances. We show\ssminor{ed} that \ssupdatedd{modeling} such systems as time-varying and using our replanning procedure \ssminor{allows} the planning of less conservative safety-guaranteed trajectories for faster goal-reaching. We demonstrate\ssminor{d} that in our replanning method, planning system \ssminor{teleportation} can be exploited to trade off between faster \ssupdatedd{goal-reaching} versus feasible navigation through tight environments. Finally, we demonstrate\ssminor{d} that our replanning procedure can be used to safely control \ssminor{the AUV} over arbitrarily long time intervals.

\ssreview{Although real-time performance is aided by the offline computation of the optimal tracking controller and the support for simple planning models, further work is required for concretely implementing online planning in Alg.~\ref{alg:planning-bounded-interval} efficiently in specific applications, as it is the main computational bottleneck. Using the AUV case study as an example, it would be interesting to analyze the time complexity of the Case 1 obstacle computation procedure in Alg.~\ref{alg:computing-planner-obstacles} or determine a more efficient procedure. Moreover, it would be  useful to improve the solver for the trajectory optimization problem. We leave these considerations to future work.}

\bibliography{references.bib, biburls.bib, extra-references.bib}
\bibliographystyle{ieeetr.bst}

\appendices
\section{Proof of Theoretical Results} \label{sec:proofs}

\begin{proof}[Proof of Prop.~\ref{prop:remain-in-teb}]
Suppose $t' \in [t,T_{\text{off}}]$. Then, \ssreview{$V(r,t)= \max_{\tau \in [t,T_{\text{off}}]} l(\xi_g^*(\tau;r,t))= \max \Big\{ \max_{\tau \in [t,t']} l(\xi_g^*(\tau;r,t)),\max_{\tau \in [t',T_{\text{off}}]} l(\xi_g^*(\tau; \xi_g^*(t';r,t),t')) \Big\}$ $\geq \max_{\tau \in [t',T_{\text{off}}]} l(\xi_g^*(\tau; \xi_g^*(t';r,t),t'))= V(\xi_g^*(t';r,t),t')$.}
\end{proof}

\begin{proof}[Proof of Lem.~\ref{lemma:obstacle-avoidance}]
Suppose $\xi_h^*(t';p,t) \in \mathcal{F}_B(t';c)$ holds. Since $p \in \mathcal{T}_p(s,t_i;c)$, $V(Ls-Mp,t) \leq c$ holds. Letting $r = Ls-Mp$, it follows from Prop. \ref{prop:remain-in-teb} that for all $t' \in [t,T_{\text{off}}]$, $V(\xi_g^*(t';r,t),t') \leq c$ holds. This can be equivalently expressed as
\begin{align}
    & V(\xi_g^*(t';r,t),t') \leq c \\
    \iff & V(L\xi_f^*(t';s,t) - M\xi_h^*(t';p,t),t') \leq c \\
    \iff & \xi_f^*(t';s,t) \in \{ s' : V(Ls'-M\xi_h^*(t';p,t),t') \leq c \} \\
    \iff & \xi_f^*(t';s,t) \in \mathcal{T}_s(t',\xi_h^*(t';p,t);c), \label{eqn:planning-theorem-remain-tube}
\end{align}
where $\mathcal{T}_s$ is defined as \begin{equation}
    \mathcal{T}_s(p,t;c) := \{s : V(Ls-Mp,t) \leq c \} \subseteq \mathcal{S} \label{def:tracking-system-tube}
\end{equation}
and is nonempty since $c \geq \underbar{V}(s,t)$.

We now aim to prove $\mathcal{T}_s(\xi_h^*(t';s,t),t';c) \subseteq B$. Suppose $s' \in \mathcal{T}_s(\xi_h^*(t';s',t),t';c)$, such that $V(Ls'-M\xi_h^*(t';p,t),t') \leq c$. We proceed by contradiction to prove that $s' \in B$. In particular, assume $s' \in B^{\mathsf{c}}$. Since $\xi_h^*(t';p,t) \in \mathcal{F}_B(t';c)$, then $V(Ls'-M\xi_h^*(t';p,t)) > c$ follows, which contradicts $V(Ls'-M\xi_h^*(t';p,t),t') \leq c$, and thus $s' \in B$ holds. Therefore, $\mathcal{T}_s(\xi_h^*(t';p,t),t';c) \subseteq B$ holds.

Since $ \xi_f^*(t';s,t) \in \mathcal{T}_s(\xi_h^*(t';p,t),t';c)$, and $\mathcal{T}_s(\xi_h^*(t';p,t),t';c) \subseteq B$, the conclusion follows.
\end{proof}

\begin{proof}[Proof of Thm.~\ref{theorem:planning-theorem}]
Suppose $t \in [t_i,t_f]$, and ${\xi}_h^*(t;p,t_i) \in \mathcal{F}_{\mathcal{C}}(t;c)$. Since \eqref{eqn:value-function-level-condition/initial-planning-system-condition} and \eqref{eqn:planning-obstacle-avoidance} hold, applying Lem. \ref{lemma:obstacle-avoidance}, it follows that ${\xi}_f^*(t;p,t_i) \in \mathcal{C}$, establishing \eqref{eqn:planning-obstacle-avoidance}. Similarly, $\xi_h^*(t;p,t_i) \in \mathcal{G}_p(t;c)$ implies ${\xi}_f^*(t;s,t_i)\in \mathcal{G}$ via Lem. \ref{lemma:obstacle-avoidance}.
\end{proof}

\begin{proof}[Proof of Cor.~\ref{cor:replanning-bounded-interval}]
Suppose $k \in \{0, \hdots, K-1\}$, and $t \in [t_k, t_{k+1})$. 

We first prove \eqref{eqn:replanning-obstacle-avoidance}. Since $\xi_h^R(t) \in \mathcal{F}_{\mathcal{C}(t_k)}(t;c_k)$, then by definition in \eqref{eqn:replanned-trajectory} and the fact that $s_k = \xi_f^R(t_k)$ (following induction), we know that $\xi_h^*(t; p_k, \xi_f^R(t_k), t_k) = \xi_h^*(t; p_k, s_k, t_k) \in \mathcal{F}_{\mathcal{C}(t_k)}(t;c_k)$. Moreover, since $c_k \geq \underbar{V}(\xi_f^R(t_k) ,t_k)$ and $p_k \in \mathcal{T}_p(\xi_f^R(t_k),t_k;c_k)$, following \eqref{eqn:planning-obstacle-avoidance} in Thm. \ref{theorem:planning-theorem}, $\xi_f^*(t; \xi_f^R(t_k), p_k, t_k) \in \mathcal{C}(t_k)=\mathcal{C}(t)$ holds for $t \in [t_k,t_{k+1})$. Finally, since
\begin{align}
    \xi_f^R(t) &= \xi_f(t; s_0, 0, u_s^R(\cdot), d^R(\cdot)) \\
    &= \xi_f(t; \xi_f^R(t_k), t_k, u_s^R(\cdot),d^R(\cdot))\\
    &= \xi_f(t; \xi_f^R(t_k), t_k, u_s^*(\cdot; L\xi_f^R(t_k) - Mp_k,t_k), \\
    &\qquad d^*(\cdot; L\xi_f^R(t_k) - Mp_k, t_k)) \\
    &=\xi_f^*(t; \xi_f^R(t_k), p_k, t_k)
\end{align}
it follows that $\xi_f^R(t) \in \mathcal{C}(t)$.

We now prove \eqref{eqn:replanning-goal-reaching} in a similar manner. Since $\xi_h^R(t) \in \mathcal{G}(t;c_k)$, then by definition in \eqref{eqn:replanned-trajectory} and $s_k = \xi_f^R(t_k)$ we find $\xi_h^*(t; p_k, \xi_f^R(t_k), t_k)=\xi_h^*(t; p_k, s_k, t_k) \in \mathcal{G}_p(t; c_k)$. Moreover, since $c_k \geq \underbar{V}(\xi_f^R(t_k) ,t_k)$ and $p_k \in \mathcal{T}_p(\xi_f^R(t_k),t_k;c_k)$, following \eqref{eqn:planning-goal-reaching} in Thm. \ref{theorem:planning-theorem}, $\xi_f^R(t) = \xi_f^*(t; \xi_f^R(t_k), p_k, t_k) \in \mathcal{G}$.
\end{proof}

\begin{proof}[Proof of Cor.~\ref{cor:replanning-periodic}]
Suppose $k \in \{0, \hdots, K-1\}$, and $t \in [t_k, t_{k+1})$.

We first prove \eqref{eqn:replanning-obstacle-avoidance-periodic}. Since $\xi_h^R(t) \in \mathcal{F}_{\mathcal{C}(t_k)}(\mathcal{M}^{ \text{f} }_{\tau}(t_k,t);c_k)$, then by definition in \eqref{eqn:replanned-trajectory-periodic} and the fact that $s_k = \xi_f^R(t_k)$ (following induction) we know that $\xi_h^*(\mathcal{M}^{ \text{f} }_{\tau}(t_k,t); p_k, \xi_f^R(t_k), \mathcal{M}^{\text{i}}_{\tau}(t_k,t))=\xi_h^*(\mathcal{M}^{ \text{f} }_{\tau}(t_k,t); p_k, s_k, \mathcal{M}^{\text{i}}_{\tau}(t_k,t))$ $ \in \mathcal{F}_{\mathcal{C}(t_k)}(\mathcal{M}^{ \text{f} }_{\tau}(t_k,t);c_k)$. Moreover, since $c_k \geq \underbar{V}(\xi_f^R(t_k) , \mathcal{M}_{\tau}^i(t_k,t_{k+1})$ and $p_k \in \mathcal{T}_p(\xi_f^R(t_k),\mathcal{M}^{\text{i}}_{\tau}(t_k,t_{k+1});c_k)$, following \eqref{eqn:planning-obstacle-avoidance} in Thm. \ref{theorem:planning-theorem}, $\xi_f^*(\mathcal{M}^{ \text{f} }_{\tau}(t_k,t); \xi_f^R(t_k), p_k, \mathcal{M}^{\text{i}}_{\tau}(t_k,t))$ $  \in \mathcal{C}(t_k)=\mathcal{C}(t)$ holds for $t \in [t_k,t_{k+1})$, where Thm. \ref{theorem:planning-theorem} can be applied since \eqref{eqn:successfully-mapped-offline-interval} implies $\mathcal{M}_{\tau}([t_k,t_{k+1}]) \subseteq [0,T_{\text{off}}]$. Finally, since
\begin{align}
    \xi_f^R(t) &= \xi_f(t; s_0, 0, u_s^R(\cdot), d^R(\cdot)) \\
    &= \xi_f(t; \xi_f^R(t_k), t_k, u_s^R(\cdot),d^R(\cdot))\\
    &= \xi_f(t; \xi_f^R(t_k), t_k, \\
    &\quad u_s^*(\mathcal{M}^{ \text{f} }_{\tau}(t_k,(\cdot)); L\xi_f^R(t_k) - Mp_k,\mathcal{M}^{\text{i}}_{\tau}(t_k,(\cdot))), \\
    &\quad d^*( \mathcal{M}^{ \text{f} }_{\tau}(t_k,(\cdot)); L\xi_f^R(t_k) - Mp_k, \mathcal{M}^{\text{i}}_{\tau}(t_k,(\cdot))))) \\
    &=\xi_f^*(\mathcal{M}^{ \text{f} }_{\tau}(t_k,t); \xi_f^R(t_k), p_k, \mathcal{M}^{\text{i}}_{\tau}(t_k,t)),
\end{align}
it follows that $\xi_f^R(t) \in \mathcal{C}(t)$.

We now prove \eqref{eqn:replanning-goal-reaching-periodic} in a similar manner. Since $\xi_h^R(\mathcal{M}^{ \text{f} }_{\tau}(t_k,t)) \in \mathcal{G}(\mathcal{M}^{ \text{f} }_{\tau}(t_k,t);c_k)$, then by definition in \eqref{eqn:replanned-trajectory-periodic} and $s_k = \xi_f^R(t_k)$ we know that $\xi_h^*(\mathcal{M}^{ \text{f} }_{\tau}(t_k,t); p_k, \xi_f^R(t_k), \mathcal{M}^{\text{i}}_{\tau}(t_k,t))=\xi_h^*(\mathcal{M}^{ \text{f} }_{\tau}(t_k,t); p_k, s_k, \mathcal{M}^{\text{i}}_{\tau}(t_k,t)) \in \mathcal{G}_p(\mathcal{M}^{ \text{f} }_{\tau}(t_k,t); c_k)$. Moreover, since $c_k \geq \underbar{V}(\xi_f^R(t_k) ,\mathcal{M}^{\text{i}}_{\tau}(t_k,t))$ and $p_k \in \mathcal{T}_p(\xi_f^R(t_k),\mathcal{M}^{\text{i}}_{\tau}(t_k,t);c_k)$, following \eqref{eqn:planning-goal-reaching} in Thm. \ref{theorem:planning-theorem}, $\xi_f^R(t) = \xi_f^*(\mathcal{M}^{ \text{f} }_{\tau}(t_k,t); \xi_f^R(t_k), p_k, \mathcal{M}^{\text{i}}_{\tau}(t_k,t)) \in \mathcal{G}$ holds.
\end{proof}




\addtolength{\textheight}{-12cm}

\end{document}